%% file: ratps-arxiv.tex
\documentclass[runningheads,envcountsame]{llncs}

\input{packages}

\input{macro}

\title{Program Synthesis for Non-Linear Real Arithmetic: Going Beyond Realizability}
\titlerunning{Program Synthesis for {\qfnra}: Going Beyond Realizability}

\author{
S.~Akshay\inst{1} \and
Supratik Chakraborty\inst{1} \and
R.~Govind\inst{2,3}
\and
Aniruddha R. Joshi\inst{4}
}

\authorrunning{S. Akshay, S. Chakraborty, R. Govind and A. R. Joshi}

\institute{
  IIT Bombay, India, 
  \email{\{akshayss, supratik\} @ cse.iitb.ac.in}\\
  \and
  The Institute of Mathematical Sciences, Chennai, India,
  \email{govind@imsc.res.in}
  \and
  Homi Bhabha National Institute, Mumbai, India
  \and
  University of California at Berkeley, Berkeley, USA, 
  \email{aniruddhajoshi@berkeley.edu}\\
}

\begin{document}

\maketitle

\begin{abstract}
  We study the problem of synthesizing programs from
  non-linear real arithmetic ({\qfnra}) specifications.  Existing
  techniques, such as syntax-guided synthesis ({\sygus}), fail to
  synthesize programs when the specification is \emph{unrealizable}.
  We argue this is unsatisfactory in many situations, and
  aim to synthesize programs from arbitrary {\qfnra} specifications,
  such that for any input, the synthesized program either produces
  outputs satisfying the specification or reports
  non-existence of any such output. To avoid rounding errors inherent
  in floating-point arithmetic, we restrict
  our programs to work on rational inputs and outputs.

  We first show that our variant of the synthesis problem is as hard
  as a long-standing open problem in number theory, and that
  synthesizing loop-free programs from arbitrary NRA specifications
  with rational inputs and outputs is impossible in general.  Second,
  we present a sound and complete synthesis algorithm for the case
  where the specification involves a single output variable.  We also
  show that for realizable specifications, a program generated by
  {\sygus} for {\qfnra} (real inputs and outputs) serves as a solution
  to our problem, where inputs and outputs are rationals. Third, we
  provide a sound (but necessarily incomplete) synthesis algorithm for
  the general case of specifications.  We have implemented our
  approach in a prototype tool called \nqs\ that solves many
  benchmarks beyond the reach of state-of-the-art SyGuS tools, even
  when we render the specifications realizable.
\end{abstract}

\input{intro}
\input{prelims}

\input{hardness}

\input{algorithm}
\input{experiments}
\input{conclusion}
    
\bibliographystyle{splncs04}
\bibliography{rat-ps}

\clearpage
\appendix

\input{app-stability}

\input{app-hardness}

\input{example}

\input{app-complete}

\input{app-benchmarks}

\end{document}

%% file: packages.tex
\usepackage{amssymb}
\usepackage{amsmath}
\usepackage{amsfonts}

\usepackage{float}
\usepackage{subcaption}
\usepackage{graphicx}
\usepackage{hyperref}
\usepackage{adjustbox}

\usepackage{authblk}

\usepackage{thmtools}
\usepackage{thm-restate}

\usepackage{hyperref}

\usepackage{cleveref}
\usepackage{wrapfig}

\usepackage{tikz}
\usetikzlibrary{arrows.meta}
\usepackage{bussproofs}
\usepackage{stackengine}
\usepackage{multirow}

\usepackage{logicproof}[2014/03/20]
\usepackage[noend]{algpseudocode}
\usepackage{algorithm,algorithmicx}
\usepackage{pifont}
\usepackage{stmaryrd}
\usepackage{pdfpages}

\usepackage{booktabs}
\usepackage{caption}
\usepackage{cite}
\usepackage{lineno}

\usepackage{tcolorbox}
\usepackage{longtable}

\newcommand{\eat}[1]{}

\definecolor{darkgreen}{HTML}{006400}

%% file: macro.tex
\newcommand{\rat}{\ensuremath{\mathbb{Q}}}
\newcommand{\intgr}{\ensuremath{\mathbb{Z}}}
\newcommand{\nat}{\ensuremath{\mathbb{N}}}
\newcommand{\real}{\ensuremath{\mathbb{R}}}
\newcommand{\intfactor}{\ensuremath{\mathsf{IntFactor}}}
\newcommand{\nstrict}{\ensuremath{\mathsf{NonStrict}}}
\newcommand{\candsol}{\ensuremath{\mathsf{CandRatRoot}}}
\newcommand{\var}{\ensuremath{\mathbf{V}}}
\newcommand{\inVar}{\ensuremath{\mathbf{X}}}
\newcommand{\invar}{\ensuremath{\mathbf{X}}}
\newcommand{\outvar}{\ensuremath{\mathbf{Y}}}
\newcommand{\XX}{\ensuremath{\mathbf{X}}}
\newcommand{\YY}{\ensuremath{\mathbf{Y}}}
\newcommand{\ZZ}{\ensuremath{\mathbf{Z}}}
\newcommand{\VV}{\ensuremath{\mathbf{V}}}
\newcommand{\AAA}{\ensuremath{\mathbf{A}}}
\newcommand{\BBB}{\ensuremath{\mathbf{B}}}
\newcommand{\ntC}{\ensuremath{\mathsf{PCnstr}}}

\newcommand{\ntP}{\ensuremath{\mathsf{PTerm}}}

\newcommand{\tvar}{\ensuremath{\mathsf{Var}}}

\newcommand{\qfnqa}{\textsf{NRA}}

\newcommand{\qfnra}{\textsf{NRA}}
\newcommand{\prob}{\ensuremath{\mathcal{RPS}_{\rat}}}
\newcommand{\probreal}{\ensuremath{\mathcal{RPS}}}

\newcommand{\HTP}{\mathcal{HTP}_{\mathbb{Q}}}

\newcommand{\gram}{\ensuremath{\mathcal{G}}}
\newcommand{\lfgram}{\ensuremath{\mathcal{G}^{\mathit{lf}}}}
\newcommand{\prog}{\ensuremath{\mathsf{Prog}}}

\newcommand{\progone}[1]{\ensuremath{\mathsf{Prog1OP}_{#1}}}
\newcommand{\proggen}[1]{\ensuremath{\mathsf{Prog}}}
\newcommand{\rprog}[1]{\ensuremath{\mathsf{RProg}_{#1}}}

\newcommand{\algqfnras}{\textsf{RationalSamplesNRA}}

\newcommand{\algqfnqa}{\textsf{SynthNQA}}

\newcommand{\cad}{\textsf{CAD}}

\newcommand{\step}{\ensuremath{\mathsf{count}}}
\newcommand{\cutoff}{\ensuremath{\mathsf{timeout}}}

\newcommand{\true}{\textbf{ \texttt{true} }}
\newcommand{\false}{\textbf{ \texttt{false} }}

\newcommand{\nqs}{\textsf{NQSynth}}

\newcommand{\zthree}{\textsf{Z3}}
\newcommand{\cvcfive}{\textsf{CVC5}}
\newcommand{\qepcad}{\textsf{QEPCAD}}
\newcommand{\modenum}{\textsf{ModEnum}}
\newcommand{\sygus}{\textsf{SyGuS}}

%% file: intro.tex
\section{Introduction}~\label{sec:intro}
Automated program synthesis, or generating programs automatically from
specifications, is arguably one of the holy grails of computer
science. While the problem has a long and storied
history~\cite{Church1962,Turing1936}, recent years have seen a surge
of interest in synthesizing programs from logical specifications in
specialized theories, viz. propositional logic, theory of bit-vectors,
(non-)linear real/integer arithmetic, and the
like~\cite{Brown01,Mccallum99,Mccallum01,Strzebonski00,nlsat-ivanovic-demoura-ijcar12,chen-morenzo-maza-jsc16,sadeghimanesh-england-2021,kremer-abraham-jsc20,abraham-davenport-england-kremer-20,KuncakMPS10,LPAR}.
One such theory, with applications spanning multiple
domains~\cite{DoratoYA97,HongLS97,Jirstrand97,LafferrierePY01,McCallum95,Weispfenning01},
is the quantifier-free theory of {\bfseries N}on-linear {\bfseries R}eal {\bfseries A}rithmetic, also called {\qfnra}.

Over the last decade, {\bfseries Sy}ntax-{\bfseries Gu}ided {\bfseries S}ynthesis, or {\sygus}, has emerged as a powerful paradigm for
program synthesis from {\qfnra} (and other first-order theory)
specifications~\cite{sygus-Alur,sygus-fmcad,sygus-fisman,sygus-acm}.
At an abstract level, a {\sygus} tool systematically searches through
a space of candidate programs guided by a grammar, relying on logical
checks to determine if a candidate program satisfies the given
specification \emph{for all possible} input values.  In many practical
applications, however, the natural specification is such that there
are some input values for which no output values satisfy the
specification.  Such specifications are called \emph{unrealizable} in
the program synthesis literature.  In such cases, the logical check
referred to above must fail for every candidate program, preventing a
{\sygus} tool from synthesizing any program at all.  This can be
highly unsatisfactory for an end-user, who is left with no program
simply because some (perhaps never-to-be-used) input values do not
admit any correct outputs. From a practical point of view, it is often
adequate to have a program that generates correct outputs whenever the
inputs admit existence of outputs satisfying the specification, and
faithfully reports the absence of correct outputs otherwise.
Unfortunately, {\sygus} tools fail to generate such programs for
unrealizable specifications.  In this paper, we propose a novel
approach to address this problem, that works regardless of whether the
specification is realizable.  Interestingly, our approach solves
problems beyond the capabilities of state-of-the-art {\sygus} tools,
even for realizable specifications. This effectively adds to the
repertoire of known techniques for program synthesis from {\qfnra}
specifications.

To illustrate how unrealizability comes in the way of synthesizing
programs using {\sygus} tools, consider the specification $0.9 \le x^2
+ y^2 \le 1$, where $x$ is a real-valued input and $y$ is a
real-valued output.  It is easy to see that this specification is
unrealizable (e.g. if $x=2$, there is no real $y$ that satisfies the
specification).  State-of-the-art {\sygus} tools like CVC5~\cite{CVC5}
therefore do not synthesize any program, and simply report this
specification as unrealizable.  Unfortunately, this may not be useful
to the end-user, who wants a program to calculate a value of $y$
satisfying the specification, whenever possible.  A work-around would
be to modify the specification and render it realizable before using a
{\sygus} tool.  For example, the above specification could be modified
to $(-1 \le x \le 1) \Rightarrow (0.9 \le x^2 + y^2 \le 1)$, where $-1
\le x \le 1$ is the \emph{weakest pre-condition} on $x$ that admits a
real $y$ satisfying $0.9 \le x^2 + y^2 \le 1$.  Violation of the weakest
pre-condition can indeed be used by the program to detect when no real
output satisfying the specification exists.  Unfortunately, finding
the weakest pre-condition is equivalent to existentially projecting
out output variables from the specification.  This is technically
involved and computationally expensive for {\qfnra} in
general~\cite{qe-nfqra-cad-lb}, and is a tall ask from the user. A
primary goal of this paper is to develop techniques that can help
users in such cases, without worrying about realizability or weakest
pre-condition calculation.

An important point that deserves attention when discussing programs
with real number arithmetic is the representation of real numbers.  It
is well known that arithmetic with fixed precision floating point
representation of real numbers can incur significant rounding errors
(see~\cite{goldberg} for an excellent exposition on this
topic). 
We discuss the challenge posed due to this problem in Appendix~\ref{app:intro}.
For several applications (viz. scientific computing, high volume
financial transactions etc.) such rounding errors are unacceptable,
and we must use arbitrary precision floating point numbers or
rationals, implemented in well-engineered libraries like
GMP~\cite{GMP2000}.  This comes at a price: software implementing
arbitrary precision floating point or rational arithmetic are
significantly slower than fixed-precision floating point arithmetic,
implemented through dedicated hardware in modern processors.  However,
this tradeoff between accuracy and performance is inevitable in
programs using real arithmetic. In our work, we wish to synthesize
programs that have \emph{zero rounding errors}.  Since rationals
strictly subsume arbitrary precision floating point numbers, we choose
to synthesize programs from {\qfnra} specifications, \emph{assuming
rational inputs and outputs}.  This yields a new variant of the
program synthesis problem from {\qfnra} specifications, that is of
independent interest.  Interestingly, this also allows us to exploit
properties of rational numbers to go beyond the capabilities of
{\sygus} tools for {\qfnra} specifications, as shown later in the
paper.

The above discussion motivates the following problem definition.  Let $\rat$ denote the set of rational numbers, and $\inVar$ and $\outvar$ denote sequences of input and output variables respectively.  We use $|\inVar|$ (resp. $|\outvar|$) to denote the count of variables in $\inVar$ (resp. $\outvar$).
\begin{tcolorbox}
  Given an {\qfnra} specification $\varphi(\inVar, \outvar)$, generate a
  program $\prog$ s.t.
  \begin{itemize}
  \item All conditions and expressions in $\prog$ are {\qfnra} formulas and
    terms
  \item For every $\AAA \in \rat^{|\inVar|}$, $\prog$ run with
    $\AAA$ as input terminates, and produces one of:
    \begin{itemize}
      \item $\bot$ (representing ``no output'') if $\forall \YY \in
        \rat^{|\YY|}\colon\neg\varphi(\AAA, \YY)$ holds, or
      \item $\BBB \in \rat^{|\outvar|}$ such that $\varphi(\AAA, \BBB)$ holds
    \end{itemize} 
  \end{itemize}
\end{tcolorbox}
\noindent We allow $\prog$ to have loops, conditional statements and
assignment statements (a formal grammar appears in
Section~\ref{sec:prelims}).  We call the above problem \emph{real
program synthesis with rational inputs and outputs}, and use $\prob$
to denote the class of all instances of this problem.
Further, we denote by $\probreal$ the variant of the above problem where $\rat$ in the problem statement is replaced by $\real$ (set of reals).

The primary contributions of our paper can now be summarized as follows:
\begin{enumerate}
\item {\it Theoretical results} (Section~\ref{sec:hardness})
\begin{enumerate}
\item We show that $\prob$ and Hilbert's tenth problem over rationals
  (a long-standing open problem \cite{eisentraeger2016easy}) are
  inter-reducible.
\item We show the impossibility of solving $\prob$ using loop-free
  programs.
  \end{enumerate}
\item {\it Algorithms} (Section~\ref{sec:algorithm})
  \begin{enumerate}
  \item We show that if $\varphi(\inVar,\outvar)$ is realizable over
    reals, any program synthesized by a {\sygus} tool for {\qfnra} is
    also a solution to $\prob$, but not the other way round.
  \item We provide a sound and complete algorithm for $\prob$ when
    $|\outvar| = 1$.
  \item When $|\outvar| > 1$, we provide a sound and (necessarily)
    incomplete procedure for $\prob$.  We can effectively detect when
    a solution generated by our algorithm solves the instance of
    $\prob$ exactly.
  \end{enumerate}
\item {\it Implementation and experiments} (Section~\ref{sec:experiments}):
  We implement our algorithm in a prototype tool called \nqs\ and use
    it to synthesize programs for a suite of {\qfnra} specifications.
    Our experiments show that \nqs\ correctly synthesizes programs for
    several non-trivial specifications, many of which are
    unrealizable.  Furthermore, our approach is significantly more
    performant than two competing approaches on many benchmarks.
\end{enumerate}

\input{related-works}

%% file: related-works.tex
\paragraph*{Related work}
Program synthesis from {\qfnra} specifications with real (not
rational) inputs and outputs is closely related to quantifier
elimination in the existential theory of reals
(ETR)~\cite{tarski1951}, and to techniques for finding real roots of
polynomials~\cite{basu,RRI-latest,Lazard2009,Vincent1973}. While
existentially quantifying all outputs from a given specification
yields the weakest pre-condition on inputs, for specifications with a
single real output, polynomial root finding techniques can be used to
find disjoint intervals for all admissible real values of the output.
However, such root finding techniques do not distinguish between
rational and irrational roots, and cannot guarantee zero rounding
errors with floating-point implementations.  There has been a long and
illustrious line of work on quantifier elimination in ETR, starting
with the seminal work of Tarski~\cite{tarski1951}, followed by
Collins' celebrated CAD algorithm~\cite{CAD,CAD2}, and several
subsequent adaptations~\cite{Brown01,mccallum93,QEPCAD}.  Similarly,
finding real roots of polynomials is an intensely studied topic in
algebraic geometry (see~\cite{basu} for an excellent exposition).  
We leverage this rich line of work in our solution, while being
cognizant of the high complexity of algorithms like CAD, and of the
limitations of polynomial root finding techniques in the rational
setting.

Earlier approaches to program synthesis from {\qfnra} specifications
can be broadly classified as template-based approaches and
syntax-guided approaches. The work
of~\cite{sketch-based,GoharshadyHMM23,GoharshadyHMM23-arxiv} present
template-based techniques for synthesizing programs by filling holes
in user-provided sketches or templates. Syntax-guided or {\sygus}
approaches~\cite{sygus-Alur,sygus-fmcad,sygus-fisman,sygus-acm} have
also proven to be very effective with state-of-the-art
competition-winning {\sygus} tools, such as CVC5~\cite{CVC5} and
DryadSynth~\cite{DryadSynth}, synthesizing non-trivial programs from a
large class of specifications, but only if the specification is
realizable.  Our work neither fits the template-based synthesis
paradigm, nor can it be viewed as an instance of {\sygus}.  Instead,
it makes careful use of a handful of generic modules for reasoning
about polynomials to synthesize programs from specifications that go
beyond the capabilities of existing techniques.

There are also works that synthesize programs from restricted real/integer arithmetic specifications. For example, ~\cite{KuncakMPS10,KuncakJournal} discuss techniques to solve program synthesis for specifications that are combinations of \emph{linear} constraints over rational/integer variables, and other SMT theories. Similarly,~\cite{LPAR} considers specifications given as non-linear constraints over \emph{bounded integers}. However, these techniques do not generalize easily to
specifications in {\qfnra} in its full generality, and hence we don't
compare against them.

%% file: prelims.tex
\section{Preliminaries and Problem statement}~\label{sec:prelims}
Let $\nat$, $\intgr$, $\rat$, $\real$ denote the set of natural numbers, 
integers, rational numbers and real numbers, respectively.
Let $\inVar = \{x_1, \ldots x_m\}$ and $\outvar = \{y_1, \ldots y_n\}$ be disjoint sets of variables, representing program inputs and outputs respectively. A \emph{relational specification}, or simply a \emph{specification}, of a program is a first-order formula $\varphi$ with free variables $\var=\inVar \cup \outvar$.  
In this paper, we focus on specifications in the
\emph{quantifier-free theory of non-linear real arithmetic} ({\qfnra}), where variables are assumed to take values from $\real$.

Thus, our specifications are Boolean combinations of polynomial inequalities over real-valued input and output variables.  However, for reasons related to rounding errors as explained in Section~\ref{sec:intro}, we will restrict our discussion to rational input and output values. 
We use lower-case letters $v, x, y, z$, possibly with subscripts, to
denote variables, while constants are represented by letters $a, b$.
Sets of variables are denoted as $\VV = \{v_1, \ldots, v_k\}$, and
similarly for $\XX, \YY, \ZZ$. Likewise, sets of constants are denoted
$\AAA = \{a_1, \ldots ,a_k\}$. A monomial over $\VV$ is a (repeated)
product of variables in $\VV$, written as $v_1^{r_1}v_2^{r_2}\cdots
v_k^{r_k}$, where $r_1, \ldots r_k \in \nat \cup \{0\}$.  A polynomial
is a rational weighted sum of monomials, i.e. $\sum_{j=1}^t a_jm_j$,
where $a_j \in \rat$ and $m_j$ is a monomial.  Polynomials over
$\VV$ are represented as $p(\VV), q(\VV), r(\VV)$ etc.
A polynomial inequality is a constraint of the form $p(\VV) \bowtie 0$, where $\bowtie ~\in \{<, >, \le, \ge\}$.  Polynomial constraints over $\VV$ are Boolean combinations of polynomial inequalities over $\VV$, and are represented as $\varphi(\VV), \psi(\VV)$ etc.
For a constant $a$ and $v \in \VV$, we use $\varphi[v \mapsto a]$ to denote the constraint obtained by substituting $a$ for $v$ in $\varphi$.
If $\AAA$ denotes the tuple of constants $(a_1, \ldots a_k)$ and $\VV$ denotes the tuple of variables $(v_1, \ldots v_k)$, we use $\varphi[\VV \mapsto \AAA]$ or simply $\varphi(\AAA)$ to denote $\big(\cdots(\varphi[v_1 \mapsto a_1])\cdots\big)[v_k \mapsto a_k]$.

\paragraph*{Problem Statement.} Given an {\qfnra} specification
$\varphi(\XX, \YY)$, our goal, as outlined in Sec.~\ref{sec:intro}, is
to synthesize a terminating program $\prog$ that takes as input
arbitrary rational values (say, $\AAA \in \rat^{|\XX|}$),
and either generates $\BBB \in\rat^{|\YY|}$ such that $\varphi(\AAA,
\BBB)$ holds, or reports that no such $\BBB$ exists in $\rat^{|\YY|}$.
In the latter case, we allow our synthesized programs to return a
special symbol $\bot$. A program satisfying this property is said to
\emph{realize} the \qfnra\ specification $\varphi(\inVar,\outvar)$.
As we show shortly, synthesizing a program that realizes a
specification in \qfnra\ is not possible in general, unless some
long-standing problems in number-theory are resolved.  Hence, we
propose to develop a synthesis algorithm that is sound but incomplete
in general.  In other words, whenever the synthesized program returns
$\BBB \in \rat^{|\YY|}$ for an input $\AAA \in \rat^{|\XX|}$, then
$\varphi(\AAA, \BBB)$ indeed holds.  However, there may exist some
other $\AAA \in \rat^{|\XX|}$ for which the synthesized program
returns $\bot$ although $\exists \YY\; \varphi(\AAA, \YY)$ holds.  For
some sub-classes of specifications, however, we can indeed obtain a
sound and complete synthesis algorithm.  Thus, for specifications in
these sub-classes, the synthesized program returns $\bot$ for an input
$\AAA$ iff $\forall \YY\;\neg\varphi(\AAA, \YY)$ holds. The formal
definition of the problem statement has already been given in the
introduction.

\paragraph*{Form of the synthesized program} 
To complete the description of the problem definition, we need to
specify the form of programs that we wish to synthesize. For this, we
give a simple context-free grammar $\gram$ in
Fig.~\ref{fig:prog-grammar} for imperative programs with variables
that take rational values.  Programs derived from this grammar use
polynomial terms (including rational constants) in assignments
(denoted \ntP) and polynomial constraints, i.e, Boolean combinations
of polynomial inequalities in conditionals (denoted \ntC).
A program returns either a special symbol $\bot$ or a tuple of
rational values of variables (denoted \tvar).
\begin{wrapfigure}[]{l}{0.65\textwidth}
\begin{center}
\begin{tabular}{rcl}
  $\mathsf{Prog}$ & ::=& $\mathsf{AssignStmt}\mid \mathsf{CondStmt}\mid \mathsf{LoopStmt}\mid$ \\
   & & $\mathsf{RetStmt}\mid \mathsf{Prog}; \mathsf{Prog}$\\
  $\mathsf{AssignStmt}$ & ::= & $\tvar \leftarrow \ntP$ \\
    $\mathsf{CondStmt}$ & ::= & \textbf{if} (\ntC) \textbf{then} $\mathsf{Prog}$ \textbf{else} $\mathsf{Prog}$\\
    $\mathsf{LoopStmt}$ & ::= & \textbf{while} (\ntC) \textbf{do} $\mathsf{Prog}$\\
  $\mathsf{RetStmt}$ & ::= & \textbf{return} $\big(\mathsf{VarTuple}\big)\mid \textbf{return}\, \bot$\\
  $\mathsf{VarTuple}$ & ::= & $\tvar\mid\mathsf{VarTuple}, \tvar$\\
\end{tabular}
\end{center}
\caption{Grammar for programs}\label{fig:prog-grammar}
\end{wrapfigure}
It is not hard to see that programs derived from this grammar are Turing powerful (e.g. counter-machines can be encoded).
Let $\prog$ be a program in the language of the grammar $\gram$, and
let $\XX$ be a tuple of variables used in $\prog$.  We use the
notation $\prog(\XX)$ to denote the fact that $\prog$ does not assign
to any variable in $\XX$, and instead uses values of these variables
as inputs.  For $\AAA \in \rat^{|\XX|}$, we also use $\prog(\AAA)$ to
denote the result (return value) of running the program $\prog$ with
variables in $\XX$ set to constants in $\AAA$. A special
class of programs is obtained by disallowing the use of the
non-terminal $\mathsf{LoopStmt}$ in the above grammar.  Specifically,
with this restriction, we can only derive \emph{loop-free}
programs. Hence, we call $\gram$ with the above restriction as
$\lfgram$ (for loop-free grammar).

Note that the above grammar, though simple, is quite expressive.  For
instance, number-theoretic algorithms often used in the context of
rational numbers, such as real root isolation yielding rational bounds for
real roots of polynomials~\cite{RRI-latest}, or
application of rational root theorem~\cite{rrt} to a find rational
root of a polynomial, can be expressed as programs derivable from our
grammar.  This can be easily checked from publicly available
pseudo-code for these algorithms (see, e.g.~\cite{RRI}
and~\cite{rrt-code}).
We will make use this convenience later in the
paper.

%% file: hardness.tex
\section{Hardness Results and Theoretical Characterizations}
\label{sec:hardness}

In this section, we state our theoretical characterization and hardness results. We start by recalling Hilbert's tenth problem~\cite{mrdp} which asks, given a Diophantine equation (a polynomial equation with integer coefficients and a finite number of unknowns), whether we can decide if the equation has a solution over integers. This problem is famously undecidable by a sequence of results, which is collectively referred to as the MRDP theorem~\cite{mrdp}. Over reals, on the other hand, this problem is decidable and an algorithm for it follows from Tarski's quantifier elimination~\cite{tarski1951}. However, over rationals, it turns out that this problem is still open. More precisely, {\em the Hilbert's tenth problem over rationals, denoted $\HTP$}, asks given a polynomial equation with integer or rational coefficients and a finite number of unknowns, whether we can decide if it has a solution over rationals. Despite significant effort in this direction~\cite{eisentraeger2016easy}, this question remains a long-standing open problem. It turns out that the problem defined in the previous section is as hard as $\HTP$. Note that while $\HTP$ is a decision problem, $\prob$ is a synthesis problem; hence the statements are suitably framed.

\begin{restatable}{theorem}{thmhardness}
  \label{thm:hardness}
  \begin{enumerate}
    \item $\prob$ is at least as hard as $\HTP$.
    \item $\HTP$ is at least as hard as $\prob$.
  \end{enumerate}
\end{restatable}
 \begin{proof}
    1. Given an instance of $\HTP$, i.e., a polynomial equation
    $\varphi(\var)$, we can build the specification $\varphi(\XX,\YY)$
    with empty inputs $\XX=\emptyset$ and output $\YY=\var$. This is
    an instance of $\prob$. Now suppose we are able to synthesize a
    terminating program $\prog(\inVar)$ for this specification. Then
    on running $\prog$ with no inputs (since $\XX = \emptyset$), if it
    returns $\bot$, this means that $\varphi(\var)$ has no solution.
    If, on the other hand, it returns $\BBB\in \rat^{|\YY|}$, then
    this is guaranteed to be a solution. Thus, $\prog$ can be used to
    build a decision procedure for $\HTP$.
  
    2. In the other direction, given an instance
    $\varphi(\inVar,\outvar)$ of $\prob$, we show below how to solve
    it using a decision procedure for $\HTP$ as a sub-routine.
    W.l.o.g. we assume $\varphi$ to be in disjunctive normal form,
    i.e. of the form $\bigvee_{i=1}^k \psi_i(\inVar,\outvar)$, where
    $\psi_i$ is a conjunction of polynomial inequalities (recall that
    negation of a polynomial inequality is itself a polynomial
    inequality). We first show that for each such $\psi_i$, we can
    obtain a single equi-satisfiable polynomial equation. This is done
    in three sub-steps.
    First, for every strict inequality of the form $x>y$ in $\psi_i$,
    we introduce a fresh rational variable, say $z'$, and replace the
    strict inequality by the conjunction $(x\ge y) \wedge
    ((x-y)z'=1)$.  Clearly, this conjunction is satisfiable iff $x >
    y$ is satisfiable.  Next, for every sub-formula of the form $x\geq
    y$ in $\psi_i$, we introduce four fresh rational variables $z_1,
    z_2, z_3, z_4$, and then replace this inequality by the equation
    $x-y=z_1^2+z_2^2+z_3^2+z_4^2$. Clearly if this equation is
    satisfied, then $x-y$ must be non-negative. By Lagrange's
    four-square theorem~\cite{Ireland1982ACI}, we also know that every
    non-negative rational is the sum of four rational squares.  So
    there are always solutions for the new variables as long as
    $x-y\geq 0$ holds.  Let $\ZZ$ denote the set of all fresh rational
    variables introduced above. Then, the above steps give us a
    conjunction of equalities, say $\psi_i'(\XX, \YY, \ZZ)$, such that
    $\psi_i(\XX,\YY)\equiv \exists \ZZ \psi_i'(\XX,\YY,\ZZ)$. Finally,
    we convert $\psi_i'$ -- a conjunction of polynomial equations --
    into a single equation. This can be easily done by observing that
    for any two polynomial equations $p=0$ and $q=0$, the formula
    $(p=0)\wedge (q=0)$ is equivalent to $p^2+q^2=0$. Thus, after the
    above sequence of steps, we obtain a single polynomial equation,
    say $\psi_i''(\XX,\YY,\ZZ)$, such that $\psi_i(\XX,\YY) \equiv
    \exists \ZZ \psi_i''(\XX,\YY,\ZZ)$ holds.  Note that
    $\psi_i''(\XX,\YY,\ZZ)$ being a single polynomial equation is also
    an instance of $\HTP$.
    
    Now, for any input $\AAA\in \rat^{|\XX|}$, if $\HTP$ on
    $\psi_i''(\AAA,\YY,\ZZ)$ answers ``No'' for all $i \in \{1, \ldots
    k\}$, then no $\YY \in \mathbb{Q}^{|\YY|}$ can satisfy
    $\varphi(\AAA,\YY) \equiv \bigvee_{i=1}^k\psi_i(\AAA,\YY)$; hence a
    program realizing $\varphi(\XX,\YY)$ must return
    $\bot$. Otherwise, if $\HTP$ returns ``Yes'' for some
    $\psi_i''(\AAA,\YY,\ZZ)$, then there exists
    $\BBB\in\rat^{|\YY|+|\ZZ|}$ s.t $\psi(\AAA,\BBB)$ holds. Since
    this is a rational vector, we can now find it by enumeration.
    This gives us a terminating program that realizes
    $\varphi(\XX,\YY)$. \qed
  \end{proof}

 The above proof requires enumeration in the final step, and therefore
 does not guarantee that the synthesized program is loop-free, even if
 we had an $\HTP$ oracle.  But there could be alternative ways of
 solving $\prob$.  We therefore ask: Is it possible to solve $\prob$
 using only loop-free programs? The answer turns out to be in the
 negative.

\begin{restatable}{theorem}{thmnoloopfree}
  \label{thm:no-loop-free}
$\prob$ cannot be solved using loop-free programs.  
\end{restatable}  
The proof, as detailed in Appendix~\ref{app:hardness}, proceeds by showing that if $\prob$ can be solved using loop-free programs, then this would imply that the theory of rational arithmetic admits quantifier elimination, and hence decidability. However, by a famous result of Robinson~\cite{robinson1959undecidability}, we know that this is not the case.

%% file: algorithm.tex
\section{Algorithmic procedures for solving \prob}\label{sec:algorithm}

Given the hardness results above, a sound and complete algorithm for
solving the \prob\ problem in its full generality is unlikely. This
motivates us to (a) consider fragments of {\qfnra} for which we can
design sound and complete synthesis algorithms, and (b) design sound
but incomplete algorithms in general, that perform well in practice.
We start with an interesting observation.

\subsection{Solutions of \probreal\ are solutions of \prob}~\label{sec:nra-synth}

We show that every program derivable from grammar \gram\ (see
Fig.~\ref{fig:prog-grammar}) that realizes a specification
$\varphi(\inVar,\outvar)$ over {\em reals}, also realizes
$\varphi(\inVar,\outvar)$ over {\em rationals}.  This immediately
suggests the following result.

\begin{restatable}{theorem}{thmnrasynthesizable}
  \label{thm:nra-synthesizable}
  Let $\mathcal{L}$ be any fragment of {\qfnra} for which there is a
  sound and complete algorithm ${\mathcal A}_{\mathcal L}$ for solving
  {\probreal}, i.e. synthesis over reals using grammar \gram.  Then
  ${\mathcal A}_{\mathcal L}$ also solves {\prob}, i.e. synthesis over
  rationals using \gram, for fragment $\mathcal{L}$.
  Moreover, there exists a fragment of \qfnra\ for which there does
  not exist any sound and complete algorithm for solving \probreal,
  but for which there is a sound and complete algorithm for solving
  \prob.
\end{restatable}
We defer the proof to Appendix~\ref{sec:app-complete}.
A consequence of Theorem~\ref{thm:nra-synthesizable} is that whenever a {\sygus} tool synthesizes a program {\prog}$(\XX)$ for a (realizable) specification $\varphi(\XX,\YY)$ \emph{over reals} using $\gram$ or any fragment of it, we can effectively use {\prog}$(\XX)$ as a solution for the
$\prob$ problem.  If $\varphi(\XX,\YY)$ is not realizable, we can
first compute the weakest pre-condition $\psi(\XX) \equiv \exists
\YY\, \varphi(\XX,\YY)$ using a quantifier elimination algorithm for
reals.  If a {\sygus} tool can then synthesize a program realizing
$\psi(\XX) \rightarrow \varphi(\XX,\YY)$ \emph{over reals}, then the
same program with a check for satisfaction of $\psi(\XX)$ also serves
as a solution for the $\prob$ problem for $\varphi(\XX,\YY)$.

\subsection{Complete synthesis for Single-Output Specifications}\label{sec:single-out}

We now propose a sound and complete synthesis procedure for solving
\prob\ when the specification has only one output variable.  Thus, the
specification is a formula $\varphi(\inVar, y)$ in $\qfnra$, where $y$
is the sole program output.  At a high level, the synthesis procedure
separates the reasoning into two cases: (i) finding solutions that
satisfy the constraints with only strict inequalities, and (ii)
finding solutions that arise only on the boundary of non-strict
inequalities. These cases are handled using two classical results from
number theory, namely \emph{Real Root Isolation} (RRI) and the
\emph{Rational Root Theorem} (RRT), along with quantifier elimination
from polynomial constraints using {\cad} or variant algorithms.  For
notational convenience, we refer to a generic quantifier elimination
procedure using {\cad} as \qepcad. Crucially, \qepcad\ is used only
while synthesizing the program; it is not invoked during execution of
the synthesized program.  We start with an overview of two key results
we rely on.

\paragraph{Real root isolation. (RRI)~\cite{RRI-Uspensky,Collins-RRI,RRI-latest}}
This refers to a family of algorithms that given a univariate
polynomial, compute a collection of intervals with rational end-points
on the real line, such that each interval contains exactly one real
root of the polynomial (if any exist).
The maximum size of an interval output by RRI, say $\varepsilon > 0$,
can be user-specified.  Using classical results on lower bounds of
root separation~\cite{basu,MinRootSep}, the value of $\varepsilon$ can
be chosen (or even iteratively reduced), so that the intervals for the
real roots obtained from RRI are pairwise disjoint.  Since the sign of
a polynomial is invariant between its real roots, gaps between the
disjoint intervals output by RRI give intervals of values where the
polynomial's sign cannot change.  Given a specification $\varphi$ that
is a conjunction of \emph{strict inequalities}, the above discussion
suggests the following procedure for computing values of $y$, given a
value of $\inVar$, such that $\varphi$ is satisfied. We first
substitute the given value of $\inVar$ in every polynomial
$p(\inVar,y)$ that appears in a constraint of the form $p(\inVar,y)
\bowtie 0 ~(\bowtie \in \{<, >\})$ in $\varphi(\inVar,y)$. This
effectively gives us a set of uni-variate polynomials in $y$.  Next,
we construct the product of all these uni-variate polynomials, and
apply RRI with an appropriate value of $\epsilon$ to obtain the
interval gaps where the product polynomial's sign is invariant.
Clearly, none of the factor polynomials change sign in any of these
intervals. By evaluating each factor polynomial at the mid-point of
each such gap interval, we can find a rational value of $y$ that
satisfies $\varphi(\inVar,y)$, if one exists.

The above argument relies on the fact that the rationals are dense in
$\mathbb{R}$. Specifically, after substituting a rational value for
$\XX$, if $r_1$ and $r_2$ are rationals such that every
real $y$ in $r_1 < y < r_2$ satisfies $\chi(\XX, y)$, then there exist
rationals between $r_1$ and $r_2$ (for instance,
$\frac{r_1+r_2}{2}$) that also satisfy $\chi(\XX, y)$.
This immediately implies that if a rational value of $\XX$ satisfies
$\exists y\, \chi(\XX,y)$ over the reals, then the same value of
$\XX$ also satisfies $\exists y\, \chi(\XX, y)$ over rationals.
We exploit this property in detecting when $\exists y\, \varphi(\XX, y)$
has a solution over rationals.

\paragraph{Rational root theorem (RRT)~\cite{rrt}.} 
The second result that we need is the rational root theorem, which can
be formulated as follows. 
\begin{theorem}[Rational Root Theorem~\cite{rrt}]\label{thm:rrt}
  Let $a_nx^n + a_{n-1}x^{n-1} + \cdots + a_1x + a_0 = 0$ be a
  polynomial equation, where $a_0, \ldots a_n \in \intgr$ and $a_0,
  a_n \neq 0$.  Then every rational solution $x = \frac{u}{v}$
  of the equation, where $u$ and $v$ are relatively prime, satisfies
  the conditions: (i) $u$ is an integer factor of $a_0$, and (ii)
  $v$ is an integer factor of $a_n$.
\end{theorem}
An immediate consequence of Theorem~\ref{thm:rrt} is that every
rational solution of the polynomial equation lies in the set
$\{\frac{u}{v} \mid u \in \intfactor(a_0), v \in \intfactor(a_n)\}$,
where $\intfactor(m)$ is the set of integer factors of $m \in
\nat$. Since each natural number has finitely many divisors, the set
of candidate rational solutions is finite. 

We use these crucial observations to synthesize a program $\progone{\varphi}(\inVar)$ that realizes $\varphi(\inVar, y)$.  Though the polynomials in our specification may use rational coefficients, an inequality involving such a polynomial can be converted to an equivalent inequality with only integer
coefficients by multiplying both sides of the inequality by the least
common multiple of all denominators of rational coefficients. Hence
Theorem~\ref{thm:rrt} applies to polynomial equations with rational coefficients as well, i.e., our specification class. 

Recall that the formula $\varphi(\inVar,y)$ may contain both strict and non-strict polynomial inequalities. 
Let $\nstrict(\varphi)$ denote the set of polynomial terms $p$ that contain the output variable $y$ and occur in non-strict inequalities of the form $p \geq 0$ or $p \leq 0$ in $\varphi$.
From $\varphi(\inVar,y)$, we can always construct the formula $\widehat{\varphi}(\inVar, y)$ obtained by converting every non-strict inequality in $\nstrict(\varphi)$ to a strict inequality.  In other words, for every $p \in  \nstrict(\varphi)$, if $p \le 0$ (resp. $p \ge 0$) is a polynomial  inequality in $\varphi$, we replace all instances of this inequality in $\varphi$ by $p < 0$ (resp. $p > 0$) to get  $\widehat{\varphi}(\inVar,y)$. We use \qepcad\  to obtain $\psi(\inVar)$ by eliminating quantifiers from  $\exists y\;\widehat{\varphi}(\inVar,y)$. The pseudocode of the program $\progone{\varphi}(\inVar)$ that realizes $\varphi(\inVar, y)$ in $\qfnra$ is given in Algorithm~\ref{algo:qfnqaone}.

\begin{algorithm}[!htbp]
  \caption{Program realizing $\varphi(\inVar,y)$ with single output}
  \label{algo:qfnqaone}
  \begin{algorithmic}[1]
    \Require Formulas $\varphi(\inVar,y)$, $\widehat{\varphi}(\inVar, y)$ and $\psi(\inVar)$ in {\qfnra}
    \Statex \Comment $\widehat{\varphi}$ (computed off-line) is obtained from $\varphi$ converting every non-strict inequality involving $y$ to strict.~~
    \Statex \Comment $\psi(\inVar)$ (computed off-line using {\qepcad}) is equivalent to  $\exists y\;\widehat{\varphi}(\inVar,y)$ ~~~~~~~~~~~~~~~
    \Procedure{$\progone{\varphi}$}{$\inVar$}
    \If{$\psi(\inVar) = \true$}
    \State $\Gamma_{\widehat{\varphi}}$ $\gets$ {\algqfnras}$\big(\widehat{\varphi}(\XX,y)\big)$
    \Comment $\widehat{\varphi}(\XX,y)$ is univariate in $y$, given $\XX$
    \ForAll {$a$ in $\Gamma_{\widehat{\varphi}}$}
    \Comment Every $a \in \Gamma_{\widehat{\varphi}}$ is a rational number
    \If{$\widehat{\varphi}(\inVar, a) ~=~ \mathsf{true}$}
       \Return $a$
       \Comment $\widehat{\varphi}(\inVar,a)$ implies $\varphi(\inVar,a)$
    \EndIf
    \EndFor
    \Else
    \ForAll {$p$ in $\nstrict(\varphi)$}
    \State $\mathsf{Cand} \gets \candsol(p = 0)$
    \ForAll {$a$ in $\mathsf{Cand}$}
    \If{$\varphi(\inVar, a)$ is $\true$}
    \Return $a$
    \EndIf
    \EndFor
    \EndFor
    \Return $\bot$
    \EndIf
    \EndProcedure
  \end{algorithmic}
\end{algorithm}
This algorithm uses a function ${\candsol}(\cdot)$ that takes as
argument a polynomial equation and returns a finite set of candidate
rational roots of the equation, as described above. Further the
algorithm uses the function {\algqfnras} which, given a formula
$\widehat{\varphi}$ on strict polynomial inequalities, returns a
finite set of rational points, such that if there is any rational
solution to $\widehat{\varphi}$, one of the returned points must
satisfy it.  This is implemented using the mid-points/end-points of
gap-intervals computed by
RRI~\cite{RRI-Uspensky,Collins-RRI,RRI-latest}. There are two cases
to consider.
\begin{itemize}
\item For a given value of $\inVar$, if $\psi(\inVar)$ evaluates to
  $\true$, we know that there is a real (and hence rational) value of
  $y$ that satisfies $\widehat{\varphi}$ (with only strict
  inequalities).  So, we first obtain the uni-variate constraint
  $\varphi'(y)$ by substituting the value of $\inVar$ in
  $\widehat{\varphi}(\inVar,y)$. Then we apply RRI on $\varphi'(y)$ to
  obtain rational points in the gap-intervals computed by RRI (see
  lines 2--3 of Algorithm~\ref{algo:qfnqaone}). Next, we check if any
  of these points satisfies $\widehat{\varphi}$, and if so
  $\progone{\varphi}$ outputs it.
  \item If $\psi(\inVar)$ evaluates to $\false$ for the given value of
    $\inVar$, then we know that there is no rational value of $y$ that
    satisfies $\widehat{\varphi}(\inVar, y)$.
  However, there may exist a rational value of $y$ that satisfies
  $\varphi(\inVar, y)$ (that includes non-strict inequalities),
  although it doesn't satisfy $\widehat{\varphi}(\inVar, y)$.  For
  this to be true, the corresponding values of $\inVar$ and $y$ must
  necessarily satisfy $p = 0$ for at least one polynomial $p \in
  \nstrict(\varphi)$. We can now apply Theorem~\ref{thm:rrt} to each
  polynomial equation $p = 0$, where $p \in
  \nstrict(\varphi)$. Effectively, we enumerate the finite set of
  candidate rational solutions of $p=0$, and check whether any of
  these satisfies $\varphi(\inVar, y)$.  If so, the corresponding
  rational solution is output by $\progone{\varphi}(\inVar)$.
  Otherwise, $\progone{\varphi}(\inVar)$ outputs $\bot$.
\end{itemize}
The proof of the following theorem follows from the construction and
reasoning detailed above.
\begin{restatable}{theorem}{thmqfnqaone}
  \label{thm:qfnqa-one}
  Given a specification $\varphi(\inVar,y)$ in {\qfnra} with a single output variable, the program $\progone{\varphi}(\inVar)$ realizes $\varphi(\inVar,y)$, thus solving $\prob$.
\end{restatable}

\begin{example}
  Consider the specification $0.9 \le x^2 + y^2 \le 1$, where $x$ is a real-valued input and $y$ is a real-valued output. 
  
  Note that here, $\nstrict(\varphi) = \{x^2 + y^2 - 0.9, x^2 + y^2 - 1\}$. 
  Algorithm~\ref{algo:qfnqaone} first computes $\widehat{\varphi}$ as $x^2 + y^2 - 0.9 > 0 \wedge x^2 + y^2 - 1 < 0$.
  Then \qepcad\ computes $\psi(x)$ as $\exists y\ \widehat{\varphi}(x,y)$, which is equivalent to $x^2 < 1$, i.e., $-1 < x < 1$.
  Thus, if we pick an $x$ from the interval $(-1,1)$, we know that there is a rational $y$ such that $\widehat{\varphi}(x,y)$ is satisfied.  
  Algorithm~\ref{algo:qfnqaone} then calls {\algqfnras} on the univariate constraint in $y$, obtained by substituting the value of $x$ in $\widehat{\varphi}(x,y)$.  This corresponds to the lines 2-5 of Algorithm~\ref{algo:qfnqaone}. 
  
  Now, if $\psi(x)$ evaluates to $\false$, Algorithm~\ref{algo:qfnqaone} considers the rational solutions of the polynomial equations obtained from by considering $p = 0$ for at least one polynomial $p \in
  \nstrict(\varphi)$, which yields $x^2 + y^2 - 0.9 = 0$ and $x^2 + y^2 - 1 = 0$.
  When $x = 1$ or $x = -1$, the second equation gives $y = 0$.
  Note that in this case, the strict formula $\widehat{\varphi}$ is not satisfiable, but the original formula $\varphi$ is satisfiable by $y = 0$. 
  This corresponds to the lines 7-10 of Algorithm~\ref{algo:qfnqaone}, where we consider the candidate rational solutions of the polynomial equations obtained from $\nstrict(\varphi)$.   
\end{example}

\subsection{Synthesis for full {\qfnqa}}~\label{sec:general}
In this section, we present a sound but incomplete procedure that
solves $\prob$ for the general case.
Given the hardness results in Section~\ref{sec:hardness}, we limit
ourselves to providing a sound procedure that nevertheless solves a
large collections of benchmarks (including unrealizable
specifications), as demonstrated by our experiments.

A naive way to solve the multiple output case would be to try to
process it one output at a time by existentially quantifying the
others. However, as explained earlier, existential quantification over
rationals is at least as hard as $\HTP$, and hence we cannot hope for
a practically efficient procedure. Instead, we propose an alternative
approach that first considers a (rational) point that satisfies the
given specification, and then uses this to decompose the synthesis
problem to a collection of single output synthesis sub-problems.  The
sound and complete procedure for the one output case described in
Section~\ref{sec:single-out} is used to solve these sub-problems
individually, and the solutions composed to incrementally synthesize
the desired program. The process is repeated until the weakest
pre-condition of the given specification is covered, or until the
algorithm times out. The main difficulty of this approach is to
maintain rationality (and correctness) of solutions, while ensuring
progress at each iteration.

\begin{algorithm}[t]
  \caption{Program for $\varphi$ in {\qfnra} solving \prob}
  \label{algo:prog-synth}
  \scriptsize	  
  \begin{algorithmic}[1]
    \Require \qfnra\ specification $\varphi(\inVar,\outvar)$, where $\outvar=\{y_1,\cdots,y_n\}$, and integer $\cutoff$.
    \Procedure{$\algqfnqa$}{$\varphi,\inVar,\outvar$}
    \State $\step \gets 0$ 
      \State $\sigma \gets \mathsf{RSolve}(\varphi,\inVar,\outvar)$ \Comment{$\sigma=(\sigma_1,\ldots,\sigma_n)\in \rat^{|\YY|}$ or $\bot$}
      \While{$(\sigma \neq \bot) \wedge (\step \leq \cutoff)$}\label{line:while} 
        \For{$i$ in $\{1,\cdots,|\outvar|\}$}\label{line:for} 
        \State $\varphi_i(\XX,y_i)\gets \varphi(\XX, \sigma_1,\ldots,\sigma_{i-1},y_i,\sigma_{i+1},\ldots \sigma_n)$
        \State $\psi_i(\XX)\gets \exists y_i \varphi_i(\XX,y_i)$
        \Comment Quantifier elimination done using \qepcad\
        \State Let $\rprog{i}(\XX)$ be the program as follows:
        \Statex \hspace*{1cm}\colorbox{lightgray}{\rprog{i}\begin{minipage}{0.3\linewidth}
          \hspace{\algorithmicindent}
          \begin{tabbing}
            \hspace{1cm}\=\hspace{1cm}\=\kill
            \hspace{2cm} $(y_1,\ldots, y_{i-1},y_{i+1},\ldots, y_n)\gets (\sigma_1, \ldots \sigma_{i-1},\sigma_{i+1},\ldots, \sigma_n)$\\
            \hspace{2cm} $y_i\gets \progone{{\varphi_i}}$
          \end{tabbing}
          \end{minipage}
          }
        \label{line:prog}   
        \EndFor
        \State $\varphi(\XX,\YY) \gets \varphi(\inVar,\outvar) \wedge \neg \psi_1(\inVar)\land \neg \psi_2(\inVar)\land \ldots \land \neg \psi_n(\inVar)$\label{line:newspec}
         \State $\sigma \gets \mathsf{RSolve}(\varphi,\inVar,\outvar)$
        \State $\step \gets \step + 1$
        \EndWhile    
        \State \Return $\proggen{\varphi}$ defined as:
        \Statex \hspace*{2cm}\colorbox{lightgray}{\proggen{\varphi}\begin{minipage}{0.4\linewidth}
          \hspace{\algorithmicindent}
          \begin{tabbing}
            \hspace{1cm}\=\hspace{1cm}\=\kill
            \hspace{2cm} \textbf{if} $\psi_1(\XX)$ \textbf{then} $\rprog{1}(\XX)$\\
            \hspace{2cm} \textbf{else if} $\psi_2(\XX)$ \textbf{then} $\rprog{2}(\XX)$\\
            \hspace{3cm} $ \ldots$\\
            \hspace{2cm} \textbf{else if} $\psi_n(\XX)$ \textbf{then} $\rprog{n}(\XX)$\\
            \hspace{2cm} \textbf{else} $\bot$
          \end{tabbing}
          \end{minipage}
          }
    \EndProcedure
  \end{algorithmic}
\end{algorithm}

The pseudo-code for our technique is given in
Algorithm~\ref{algo:prog-synth}, where snippets of the program being
synthesized are represented in grey boxes.  These are not executed
during the run of the synthesis algorithm.
We start by querying an {\qfnra} oracle for a satisfying assignment
(or solution) for $\varphi(\XX, \YY)$. The function $\mathsf{RSolve}$
returns a rational satisfying assignment projected to the output
variables, if there exists such a valuation; otherwise it returns
$\bot$. 
Given the rational vector of values $\sigma$, for each $y_i$ in $\YY$, we now consider the formula $\varphi_i(\XX,y_i)$ obtained by fixing rational values for all output variables other than $y_i$. This specification is now a single output specification, and we use the procedure described in Section~\ref{sec:single-out} to obtain the program $\progone{\varphi_i}$. 
Then by substituting the values obtained from $\sigma$ for all output variables other than $y_i$, we obtain $\rprog{i}(\XX)$, where for the $y_i^{th}$ output, we will use $\progone{\varphi_i}$.

Let $\psi_i$ be the weakest pre-condition (computed using \qepcad) for the
one output synthesis problem in the $i^{th}$ iteration.  When we quit the for loop in line 5, we conjunct $\neg \psi_i$ with $\varphi$, effectively
obtaining a new specification $\varphi$ (in line~\ref{line:newspec}) which has strictly fewer satisfying assignments, thus ensuring progress.

The while loop terminates by either a \emph{time out}, i.e., if $(\step \leq \cutoff)$ in line~\ref{line:while} is violated, or by running out of rational solutions to $\varphi$. At termination of the program, we return the program that has the form of a decision list,  which results in a sound program.
In addition, we can also provide a guarantee of partial completeness of the procedure.  That is, if the procedure terminates before timing out, then we know that the there is no way to satisfy the specification while violating all the (partial) pre-conditions obtained so far.  Hence, the program synthesized by the procedure realizes the post-condition.

In general, a query to the {\qfnra} solver in line 3 of
Algorithm~\ref{algo:prog-synth} may not return a vector of all
rational values.  If the $i^{th}$ component of the returned vector is
not rational, we skip the $i^{th}$ iteration of the for loop in line
5.  If no component of the returned vector is rational, we query the
solver again until we get a rational solution.

The following guarantee about our algorithm follows from the construction and discussion above.
\begin{theorem}\label{thm:qfnqa-soundness}
  Algorithm~\ref{algo:prog-synth} accepts as input specification $\varphi(\inVar,\outvar)$ and terminates synthesizing a program such that (1) if $\prog
  (\inVar)$ returns $\BBB \in \rat^{\outvar}$, then $\varphi(\inVar, \outvar \mapsto \BBB)$ holds, and (2) if the algorithm terminates before timeout, then $\prog(\inVar)$ solves $\prob$ by realizing specification $\varphi(\inVar,\outvar)$.
\end{theorem}

A detailed example (with multiple output variables), showing how repeated invocations of the single-output procedure synthesize values for several output variables, is given in Appendix~\ref{app:synth-program}.

%% file: experiments.tex
\section{Implementation and Experiments}
\label{sec:experiments}
We have implemented Algorithm~\ref{algo:prog-synth}  in a prototype
tool called {\nqs}. 
Our tool outputs a Python implementation of the synthesized program,
so that it is executable off-the-shelf. We use the SMT solver
\zthree~\cite{Z3} for $\mathsf{RSolve}$ that finds satisfying
assignments of polynomial constraints (steps 3, 10 of
Algorithm~\ref{algo:prog-synth}), and the specific tool
\qepcad~\cite{QEPCAD} -- a state-of-the-art open source quantifier elimination
engine for reals in SageMath~\cite{sagemath} -- for quantifier
elimination over reals (step 7 of the algorithm).  The
Python implementation of the synthesized program uses {\tt sympy}
modules for specific operations (viz. multiplication and constant
substitution) on polynomials, and also for real-root isolation and for
finding integer divisors when applying the Rational Root Theorem (see
Section~\ref{sec:single-out}).
Our tool, along with the benchmarks used in this paper, is available and can be downloaded from \href{https://github.com/anirjoshi/NQSynth}{https://github.com/anirjoshi/NQSynth}. 

Since {\sygus} tools don't synthesize programs for unrealizable specifications, we implemented an improvised {\sygus} tool for purposes of comparison. This tool first invokes \qepcad\ to compute a quantifier-free formula $\psi(\invar)$ equivalent to $\exists \outvar\,\varphi(\invar,\outvar)$ over reals.  It then feeds
$\psi(\invar) \rightarrow \varphi(\invar,\outvar)$ as a (realizable) specification over reals to a state-of-the-art {\sygus} solver, viz. \cvcfive~\cite{CVC5}.  We call this tool {\qepcad}+{\cvcfive}. We also implemented a third tool called {\modenum}, inspired by a model-enumeration based quantifier elimination technique for linear real arithmetic, due to Monniaux~\cite{Monniaux10}.  This tool starts with $\varphi(\invar,\outvar)$ and uses \zthree\ to find a real-valued model for $\varphi$.  If all output variables are assigned rational values, {\modenum} substitutes the values of all output variables in
$\varphi$ to obtain a pre-condition $\psi(\invar)$. We also synthesize
a program fragment (an {\tt if-then} statement) that assigns the
output variables values as in the model, whenever $\psi(\invar)$
holds.  The specification is then refined to $\varphi(\invar,\outvar)
\wedge \neg \psi(\invar)$ and the entire process repeated.  The
iterations continue until either the (refined) specification becomes
unsatisfiable, or until timeout.
We depict the program synthesized by \nqs\ for a given specification in  Appendix~\ref{app:synth-program}, along with an explanation of the steps of the procedure synthesizing the program.

In our implementations, when \zthree\ is invoked to find a model of
$\varphi(\invar,\outvar)$ in {\tt QF\_NRA}, it may return a model with
some output variables assigned irrational values.  Both
Algorithm~\ref{algo:prog-synth} and {\modenum} effectively ignore
these irrational values and block this model in the next invocation of
{\zthree}.

  \begin{figure}[t]
    \begin{center}
    \begin{tabular}{cc|cc}
      \includegraphics[scale=0.35]{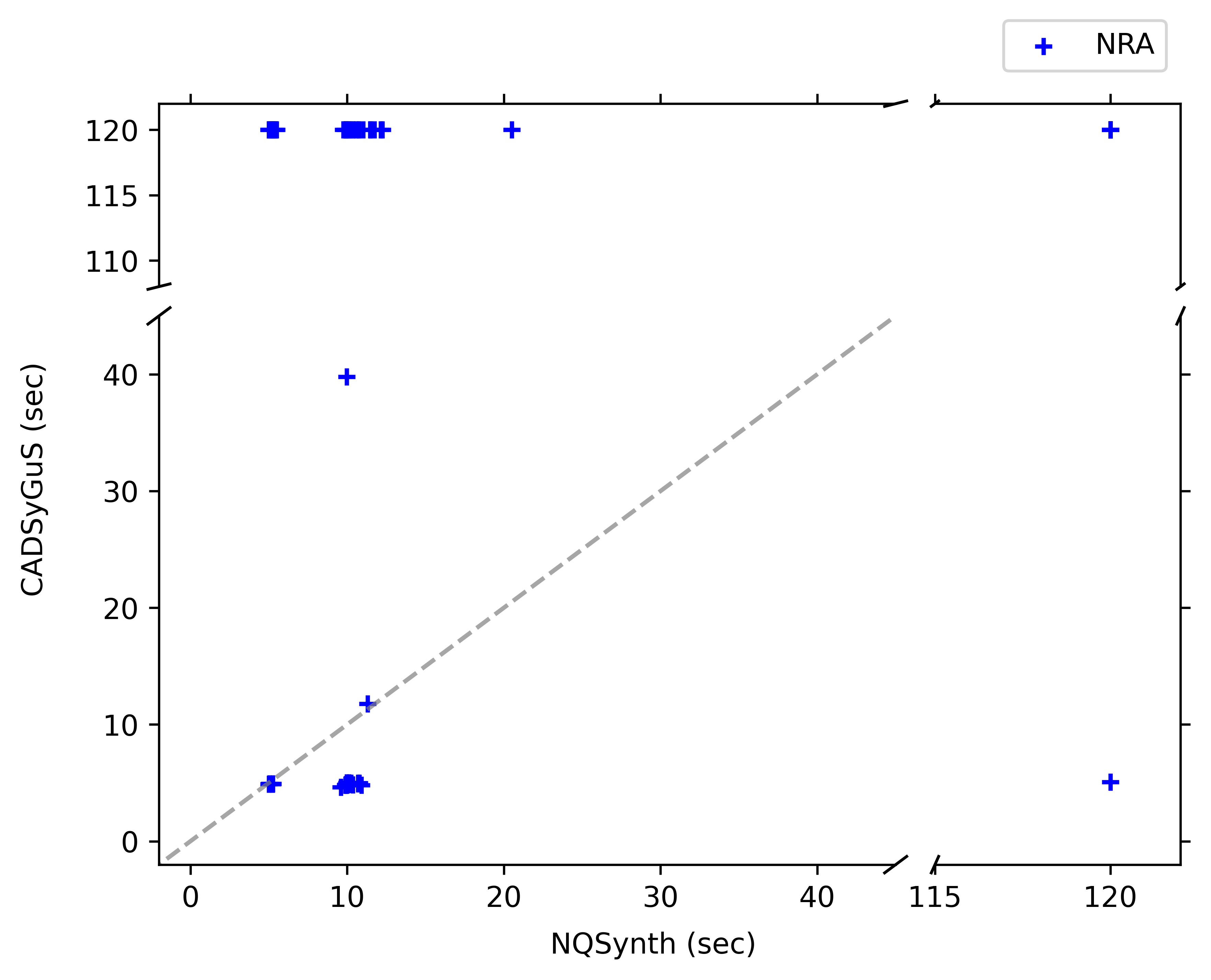} 
      &&&
      \includegraphics[scale=0.35]{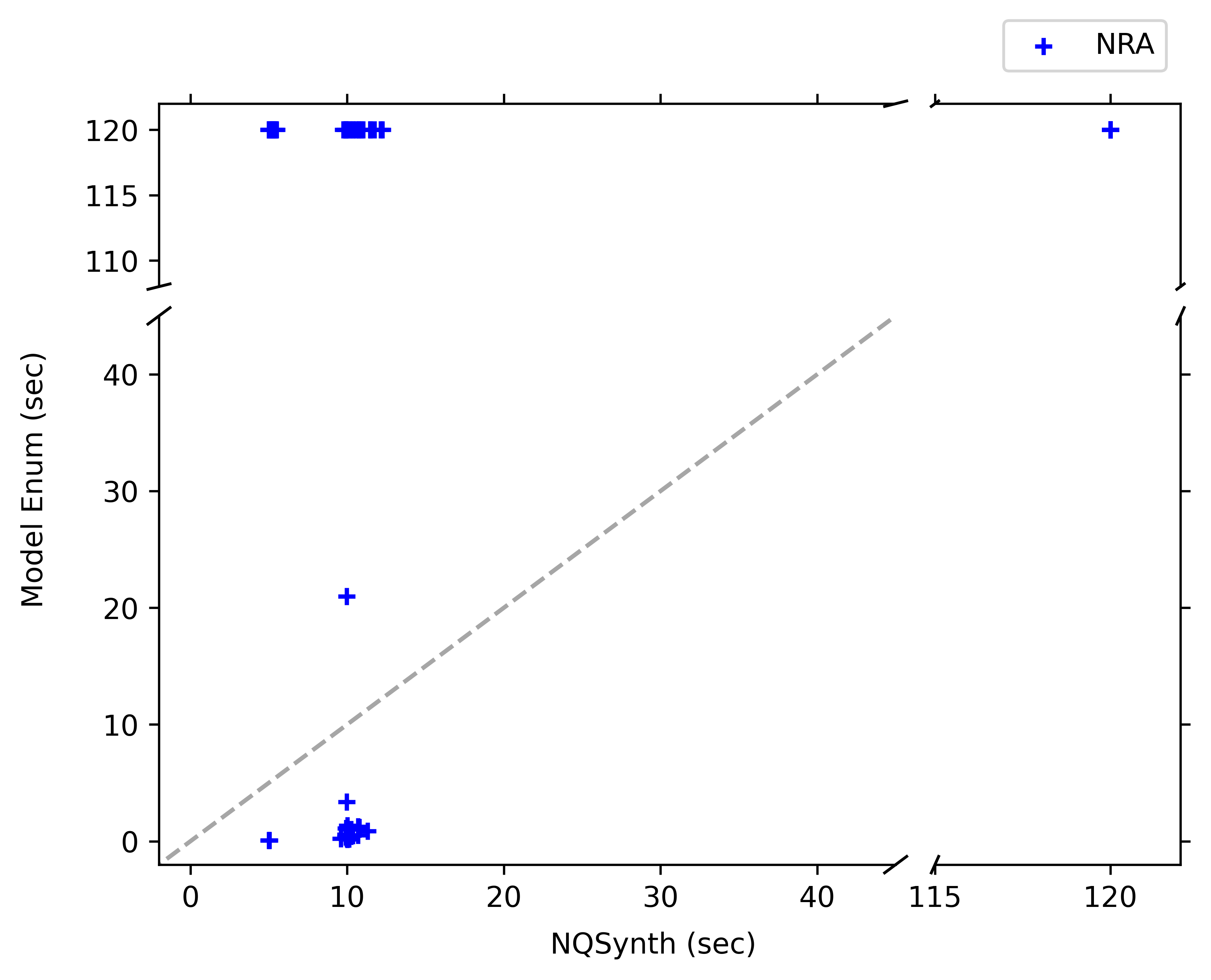} 
      \\
      \includegraphics[scale=0.35]{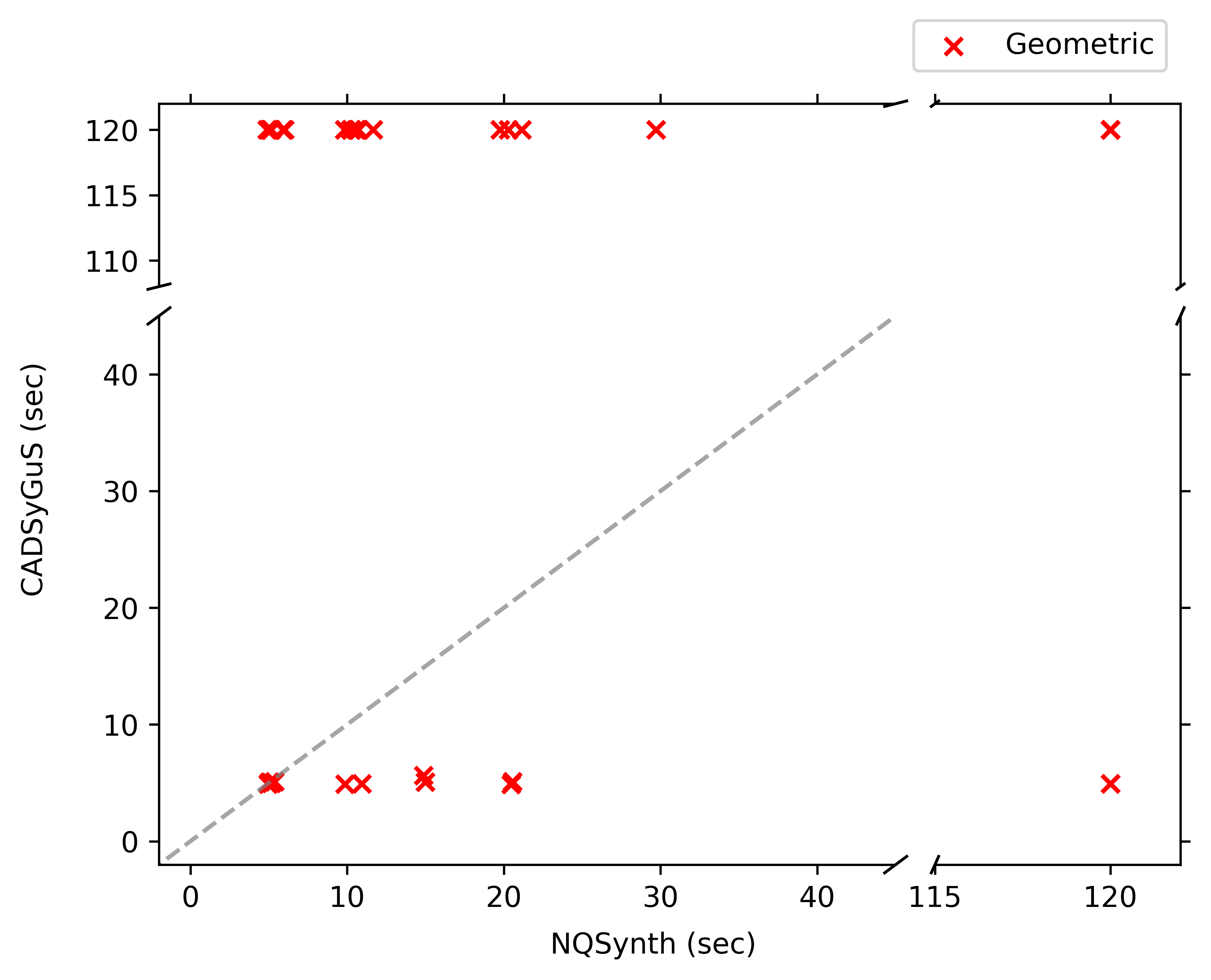} 
      &&&
      \includegraphics[scale=0.35]{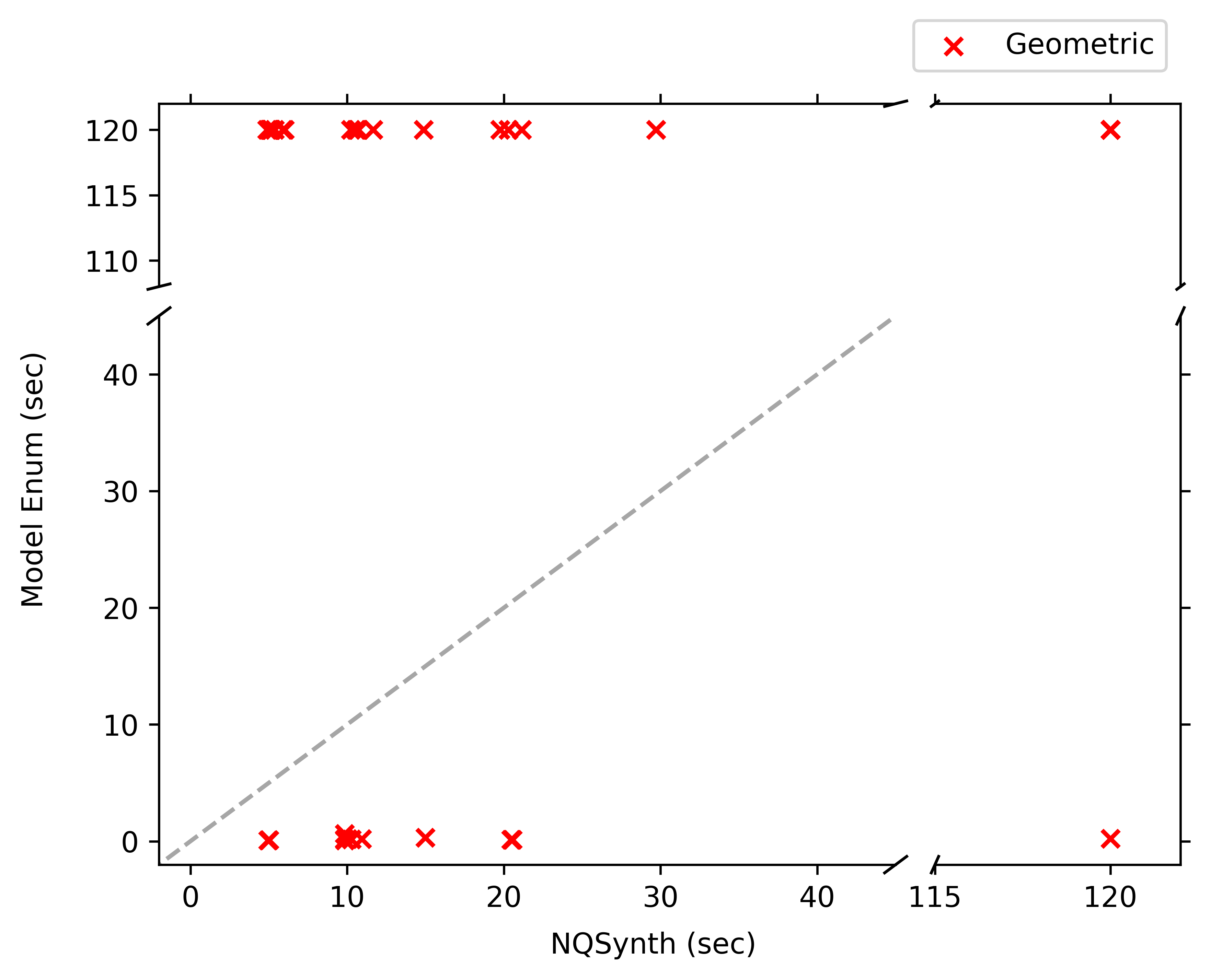} 
      \\
      (a) &&& (b)
    \end{tabular}
    \end{center}
    \caption{Run-time scatter-plots of \nqs\ with \qepcad+\cvcfive, and \modenum}    
    \label{fig:scatter-plots}
  \end{figure}

\vspace*{0.1cm}
\noindent {\emph{\bfseries Benchmarks: }} We used two sets of
  benchmarks for our evaluation.  The first set, henceforth called
  ``NRA'', are {\tt QF\_NRA} benchmarks in the SMT-LIB 2024
  repository~\cite{SMTLIB}. Since all inputs and outputs in {\prob}
  are assumed to be rational, we considered all variables in these
  benchmarks to have the rational type.  We filtered out unsatisfiable
  benchmarks, and also benchmarks with terms in mixed theories.
  During the course of experiments, we also noticed that
  \qepcad\ encounters resource constraints for large inputs involving
  $> 6$ variables.  Hence, we chose to evaluate our tool on $83$
  adapted benchmarks from the \texttt{meti-tarski/Arthan},
  \texttt{meti-tarski/asin}, and \texttt{zankl} families in the
  SMT-LIB 2024 {\tt QF\_NRA} suite.  Since {\tt QF\_NRA} benchmarks do
  not come with default listing of input/output variables, we
  partitioned the variables in each family to ensure that the
  resulting specifications (for synthesis) had non-linear constraints
  involving output variables.  Specifically, for formulas in the
  \texttt{meti-tarski/Arthan} family (approximate trigonometric
  constraints), we used {\tt skoSINS} and {\tt skoCOSS} as output
  variables.  For the \texttt{meti-tarski/asin} family, we used {\tt
    skoSP} and {\tt skoSM} as output variables, and for the
  \texttt{zankl} family, we used {\tt b} as the output variable. The
  set of resulting 83 benchmarks is listed in
  Appendix~\ref{app:expt-results}. To make the benchmarks challenging,
  we also replaced each equality constraint with a more general
  relation (of which equality is a special case).  Specifically, we
  introduced a new input variable $\delta$, and replaced all
  (sub-)formulas of the form $\textit{term}_1=\textit{term}_2$ by
  $(\delta \ge 0) \wedge (-\delta \le
  \textit{term}_1-\textit{term}_2\leq \delta)$.  
Our second set of benchmarks, called ``Geometric'', consists of $51$ Non-Linear Real Arithmetic (NRA) specifications motivated by problems in 2- and 3-dimensional geometry. These benchmarks model spatial relationships and bounding problems, such as intersecting circles, annular regions, and spheres with varying or fixed radii. The suite additionally also contains a broad range of non-linear behavior by including mixed-degree curves, parabolas, hyperbolas, and higher-power polynomial boundaries (ranging up to degree-100). While these
  benchmarks are easy to describe, synthesizing rational programs for
  them turns out to be difficult in most cases.

\vspace*{0.1cm}
\noindent {\emph{\bfseries Evaluation results: }} All our experiments
  were conducted on a MacBook Pro 2023 with Apple M2 Pro processor and
  32 GB main memory.  A timeout of 120 seconds was set for each
  technique. We say a technique \emph{solves} a
  benchmark if it synthesizes a program that \emph{realizes} the
  corresponding specification over rationals before timing out.  When
  reporting time taken by a tool for a set of benchmarks, we report
  (minimum, average, maximum) triples. The PAR2 (Penalized Average
  Runtime 2) score is a widely used metric (smaller is better) that
  computes the average runtime by penalizing each timeout with twice
  the timeout value.

  A comparison of the relative performance of \nqs,
  \qepcad+\cvcfive\ and \modenum\ on the two benchmark sets is given
  in Table~\ref{tab:comp}. Scatter plots of the times taken on
  individual benchmarks are shown in Fig.~\ref{fig:scatter-plots}a
  (\nqs\ vs \qepcad+\sygus) and Fig.~\ref{fig:scatter-plots}b
  (\nqs\ vs \modenum). There is only 1 NRA benchmark that \nqs\ cannot
  solve in 120s that {\qepcad}+{\cvcfive} solves in 5.07s.  On the
  other hand, there are 34 NRA benchmarks that {\qepcad}+{\cvcfive}
  can't solve in 120s, but all of these are solved by \nqs\ within
  25s. There is no NRA benchmark that \nqs\ cannot solve but \modenum\ can
  within 120s.\\
  \begin{table}[!tbp]
    \centering
    \resizebox{\textwidth}{!}{
    \begin{tabular}{|c|c|c|c|c|c|}
      \hline
      Benchmark Set &        & \nqs\   & \qepcad\ + \cvcfive\ &  \modenum\ \\
      \hline
      \hline
      NRA        & Solved &  59 & 26 & 23 \\
      \cline{2-5}
      Total: 83 & Time (s) & (4.998, 9.096, 20.512) & (4.649, 6.560, 39.785) & (0.067, 3.529, 45.784)\\
      \cline{2-5}
      & PAR2 & 75.864 & 166.856 & 174.472 \\
      \hline
      Geometric  & Solved & 33 &16 & 11\\
      \cline{2-5}
      Total: 51 & Time & (4.890, 10.436, 29.703) &(4.879, 5.030, 5.623)& (0.083, 0.205, 0.628)\\
      \cline{2-5}
      & PAR2 &91.459 &166.284 & 188.280\\
      \hline
    \end{tabular}
    }
    \caption{Comparison of \nqs, \qepcad+\cvcfive, and \modenum}    
    \label{tab:comp}
  \end{table}

  Our results clearly show the effectiveness \nqs\ over
  ({\qepcad}+{\cvcfive}) and {\modenum}.  On further investigation, we
  found that quantifier elimination using {\qepcad} in step 7 of
  Algorithm~\ref{algo:prog-synth} was the sole reason for phase 1
  timing out.  In 8 out of 24 timeout cases in Set 1, it was the
  first quantifier elimination itself that timed out.  In the
  remaining cases, the tool timed out during quantifier elimination
  while iterating in the outer {\tt while} loop of
  Algorithm~\ref{algo:prog-synth}, having generated between 18 to 24
  pre-conditions (whose disjunction was not yet the weakest
  pre-condition).  In the case of Geometric benchmarks, out of 18
  timeout cases for Algorithm~\ref{algo:prog-synth}, 3 timed out in
  the first quantifier elimination call.  In the other 15 cases, the
  tool timed out after having generated between 2 and 25
  pre-conditions, that did not yield the weakest pre-condition yet.
  Note that in all cases where at least one pre-condition was
  generated, Algorithm~\ref{algo:prog-synth} outputs a sound program,
  even though it timed out.  
  Interestingly, none of the invocations of \zthree\ resulted in any irrational model being generated.  
  Hence, there were no ``wasted'' models, as per the check
  in line 8 of Algorithm~\ref{algo:prog-synth}.

  Detailed data for our experiments on both sets of benchmarks are in
  the Appendix~\ref{app:expt-results}.  Overall, our experiments on
  both sets of benchmarks establish the superiority of
  Algorithm~\ref{algo:prog-synth}, and justify its inclusion as phase
  1 of {\nqs}.

To assess the impact of the underlying quantifier elimination engine, we evaluated our tool on a representative subset of 15 benchmarks (10 standard NRA and 5 synthetic geometric instances) by varying the backend to MAPLE and Mathematica. This subset was chosen to cover diverse structural constraints and includes a mix of both solvable and unsolved instances.

Replacing the default Sage-\qepcad\ backend with MAPLE and Mathematica changed the synthesis outcome for one instance. Specifically, the benchmark \texttt{Arthan\_1A\_Arthan1A-chunk-0017.smt2}, which timed out under Sage-\qepcad\ and MAPLE, was successfully solved by Mathematica in 23.91 seconds. For all other benchmarks, the solvability status remained unchanged across all three backends.
For the instances solved by all three engines, switching to MAPLE or Mathematica reduced execution time by an order of magnitude. For example, the NRA benchmark \texttt{zankl\_gen-04.smt2} was solved in 0.18 seconds using MAPLE and 0.17 seconds using Mathematica, compared to 5.49 seconds with Sage-\qepcad\. Because the remaining harder instances consistently timed out across all three engines, the bottleneck for these specific problems likely lies in the overall synthesis procedure rather than a particular quantifier elimination engine, though this requires further investigation.

Detailed total execution times for each benchmark across the three quantifier elimination engines are reported in Table~\ref{tab:qe_runtimes} in Appendix~\ref{app:expt-results}.

%% file: conclusion.tex
\section{Conclusion}
\label{sec:conclusion}
\vspace{-0.1in}

In this paper, we addressed rational program synthesis from polynomial specifications.
We investigated theoretical aspects of the problem, loop-free conditions and links to quantifier elimination.
We then developed algorithms for synthesis, sound and complete for subclasses and an {\qfnra}-oracle based sound algorithm in general. Our prototype implementation already shows promise of the approach in synthesizing rational programs for polynomial specification.  As future work, we would like to know when a particular polynomial specification has a loop-free program, if we can effectively synthesize it. 
Finally, a practical bottleneck in this work is the dependence on {\qepcad} and removing this would be a path towards further applicability. 

%% file: app-stability.tex
\section{Appendix for Section~\ref{sec:intro}}
\label{app:intro}

\subsection{Stability Challenge for reals}
\label{sec:stability}

Consider the specification $\varphi(x,y)=16 \leq x^2+y^2\leq 25$, where $y$ is an input and $x$ is an output, and the goal is to synthesize a program that for every value of $x$ provides a rational value of $y$, when it exists.
The pictorial representation of this specification is given in Figure~\ref{fig:simple}.
\begin{figure}
  \centering
  \begin{tikzpicture}[
    line cap=round,
    line join=round,
    >=Triangle,
    myaxis/.style={->,thick}
  ]
  \fill [red,even odd rule] (0,0) circle[radius=2cm] circle[radius=1.5cm];
  
  \foreach \radius in {1.5,2}
    \draw [thick] (0,0) circle[radius=\radius cm];
  
  \begin{scope}[
    every node/.append style={
      circle,
      font=\scriptsize,
      inner sep=1pt
  }]
  
     \draw [thick] (0,0) -- (40:1.5cm) node[anchor=220] {$4$};
     \draw [thick] (0,0) -- (55:2cm) node[anchor=235] {$5$};
  \end{scope}
  
  \draw[myaxis] (-2.66,0) -- (2.66,0) node[anchor=90]{$x$};      
  \draw[myaxis] (0,-2.66) -- (0,2.66)node[anchor=0]{$y$};         
  
  \end{tikzpicture}
  \caption{A simple specification}
  \label{fig:simple}  
\end{figure}
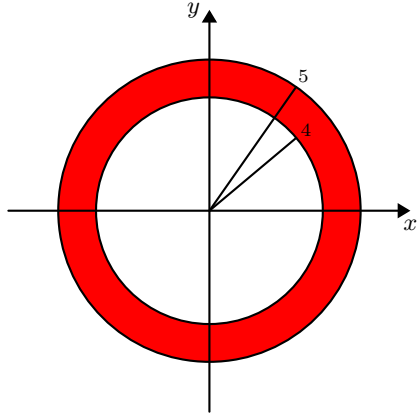

From a practical usability point of view, even close rational approximations to real algebraic numbers may lead to significant errors when evaluating polynomials.  To see why this is so, consider the conjunction of the following polynomial
constraints on input variable $x$ and output variables $y_1$ and $y_2$:
\begin{align}
  y_2^{20} - y_1^2.y_2^{19} + c_{18}.y_2^{18} + \cdots + c_1.y_2 + c_0 = 0, \label{ex:0}\\
  y_2 > 19, y_2 \leq 20, y_1^2 = x \label{ex:1}
\end{align}
where $c_{i}$ is the integer coefficient of $y_2^i$ in the expansion
of the Wilkinson's polynomial $\Pi_{j=1}^{20}\big(y_2 -
j\big)$ for all $i \in \{0, \ldots 18\}$.  It is well
known (see for example~\cite{Wilkinson,Mosier}) that the coefficient of
$y_2^{19}$ in the Wilkinson's polynomial is $-210$, but if this
coefficient is altered even slightly, there are significant changes
that happen to the roots of the Wilkinson's polynomial.  For example,
changing the coefficient of $y_2^{19}$ from $-210$ to $-210 - 2^{-23}$
in $\Pi_{j=1}^{20}\big(y_2 - j\big)$ makes $20$ a non-root of the
polynomial, and in fact, makes the polynomial evaluate to $-6.25\times
10^{17}$ when $y_2 = 20$. Thus, although the polynomial constraints
shown in (\ref{ex:0}) and (\ref{ex:1}) admit the (real algebraic)
solution $y_1 = \sqrt{210}, y_2 = 20$ when $x = 210$, if we
approximate $\sqrt{210}$ by the rational $14.4913767503$ (an error of
less than $0.0000001$), then $y_1 = 14.4913767503, y_2 = 20$ is no
longer a solution of the polynomial constraints in (\ref{ex:0}) and
(\ref{ex:1}).  In fact, it causes the polynomial in (\ref{ex:0}) to
evaluate to $-6.25\times 10^{17}$ -- a number much smaller than $0$.
Thus, {\em even close rational approximations can cause a specification that is satisfied when using the real algebraic value to become violated when using the rational approximation}.

%% file: app-hardness.tex
\section{Appendix for Section~\ref{sec:hardness}}
\label{app:hardness}

\thmnoloopfree*

\begin{proof}
  We will show that if $\prob$ can be solved using loop-free programs then this would imply that the theory of rational arithmetic admits quantifier elimination.  Consider $\varphi(\XX,\YY)$ an instance of $\prob$. By assumption we have a loop-free program $\prog(\XX)$ which solves $\prob$. This program has a tree structure where its leaves, i.e., return statements are either polynomials or $\bot$. Let us say that the polynomials at the leaves are denoted $p_1,\ldots p_r$ all over $\XX$. Then we know that for any input value of $\XX$, either there is no $\YY$ such that $\varphi(\XX,\YY)$ holds, in which case $\prog$ returns $\bot$ and otherwise $\prog$ returns some value obtained by evaluating one of the polynomials $p_1,\ldots p_r$. Consider $\psi(\XX):= \varphi(\XX,p_1)\vee \ldots \vee \varphi(\XX,p_r)$. For any $\AAA\in\rat^{|\XX|}$, if $\psi(\AAA)=1$, then $p_k(\AAA)=\BBB$ for some $p_k$ and $\BBB\in\rat^{|\YY|}$ such that $\varphi(\AAA,\BBB)=1$, i.e., $\exists \YY \varphi(\AAA,\YY)=1$. And conversely, if $\exists \YY\varphi(\AAA,\YY)=1$ for some value $\AAA$ of inputs, then by definition of the loop-free program, we obtain that $\varphi(\AAA,p_k(\AAA))=1$ for some $k$, i.e., $\psi(\AAA)=1$. Thus, $\psi(\XX)\equiv \exists \YY \varphi(\XX,\YY)$, which implies the desired claim, i.e., the theory of non-linear rational arithmetic admits quantifier elimination, i.e., performing quantifier elimination for the theory of non-linear rational arithmetic. However, by the famous result of Robinson~\cite{robinson1959undecidability}, we know that this is in fact undecidable.
  Hence we obtain that our desired result.\qed
\end{proof}

%% file: example.tex
\clearpage
\section{Example Synthesized Program}
\label{app:synth-program}

\begin{algorithm}[H]
\scriptsize
\setlength{\abovecaptionskip}{2pt}
\setlength{\belowcaptionskip}{2pt}
\caption{Synthesized program realizing $\phi(x,y,z)$}
\label{alg:synth-program-example}
\begin{algorithmic}[1]
\Require Post-condition
$\phi(x,y,z) :=
(0 \geq x^2 + \frac{y^2}{9} + \frac{z^2}{16} - 1)
\land
(0 > -\frac{x^2}{9} - \frac{y^2}{16} - z^2 + 4z - 2)$
\Require Input variable $\hat{x}$
\Require Output variables $\hat{y},\hat{z}$ such that $\phi(\hat{x},\hat{y},\hat{z})$ is satisfied

\Procedure{Main}{}
    \State $\hat{x} \gets \text{\textsc{Input}}(\text{``enter numerator and denominator of } \hat{x}\text{:''})$

    \State \textbf{/* Evaluate pre-condition 0 */}
    \If{$16\hat{x}^2 - 7 \leq 0$}
        \State $\hat{z} \gets -3$; $\hat{y} \gets \lambda$
        \State $\phi_{uv}(\lambda) \gets \phi(\hat{x}, \lambda, -3)$
        \State $\lambda_{\mathrm{val}} \gets \progone{\phi_{uv}}(\hat{x})$
        \If{$\lambda_{\mathrm{val}} \neq \bot$}
            \State \Return $\hat{y}=\lambda_{\mathrm{val}}$, $\hat{z}=-3$
        \Else
            \State \Return $\bot$
        \EndIf
    \EndIf

    \State \textbf{/* Evaluate pre-condition 1 */}
    \If{$36\hat{x}^2 - 35 \leq 0$}
        \State $\hat{y} \gets \frac{1}{2}$; $\hat{z} \gets \lambda$
        \State $\phi_{uv}(\lambda) \gets \phi(\hat{x}, \frac{1}{2}, \lambda)$
        \State $\lambda_{\mathrm{val}} \gets \progone{\phi_{uv}}(\hat{x})$
        \If{$\lambda_{\mathrm{val}} \neq \bot$}
            \State \Return $\hat{y}=\frac{1}{2}$, $\hat{z}=\lambda_{\mathrm{val}}$
        \Else
            \State \Return $\bot$
        \EndIf
    \EndIf

    \State \textbf{/* Evaluate pre-condition 2 */}
    \If{$1024\hat{x}^2 - 1023 \leq 0$}
        \State $\hat{y} \gets \lambda$; $\hat{z} \gets \frac{1}{8}$
        \State $\phi_{uv}(\lambda) \gets \phi(\hat{x}, \lambda, \frac{1}{8})$
        \State $\lambda_{\mathrm{val}} \gets \progone{\phi_{uv}}(\hat{x})$
        \If{$\lambda_{\mathrm{val}} \neq \bot$}
            \State \Return $\hat{y}=\lambda_{\mathrm{val}}$, $\hat{z}=\frac{1}{8}$
        \Else
            \State \Return $\bot$
        \EndIf
    \EndIf

    \State \textbf{/* Evaluate pre-condition 3 */}
    \If{$(\hat{x}+1 \geq 0) \land (\hat{x}-1 \leq 0)$}
        \State $\hat{y} \gets 0$; $\hat{z} \gets \lambda$
        \State $\phi_{uv}(\lambda) \gets \phi(\hat{x}, 0, \lambda)$
        \State $\lambda_{\mathrm{val}} \gets \progone{\phi_{uv}}(\hat{x})$
        \If{$\lambda_{\mathrm{val}} \neq \bot$}
            \State \Return $\hat{y}=0$, $\hat{z}=\lambda_{\mathrm{val}}$
        \Else
            \State \Return $\bot$
        \EndIf
    \EndIf

    \State \Return $\bot$ \Comment{$\boldsymbol{\nexists y,z\;\phi(\hat{x},y,z)}$}
\EndProcedure
\end{algorithmic}
\end{algorithm}

We illustrate Algorithm~\ref{algo:prog-synth} on the following specification with one input $x$ and two outputs $y,z$:

\[
\varphi(x,y,z) :=
\left(0 \ge x^2+\frac{y^2}{9}+\frac{z^2}{16}-1\right)
\wedge
\left(0 > -\frac{x^2}{9}-\frac{y^2}{16}-z^2+4z-2\right).
\]

Algorithm~\ref{alg:synth-program-example} depicts the program generated by Algorithm~\ref{algo:prog-synth} for a specification with one input variable $x$ and two output variables $y, z$. 
It is important to note that the branches of the generated program are not guessed directly.  They arise from successive iterations of Algorithm~\ref{algo:prog-synth}: in each iteration,
\textsc{RSolve} is invoked on the current residual specification, a rational
point is obtained, and all but one of the output variables are fixed according
to this model.  This yields a single-output synthesis problem, which is then
handled by Algorithm~\ref{algo:qfnqaone}.  The guard of the corresponding branch
is the precondition obtained by eliminating the remaining output variable from
this single-output instance.

Next, we describe how the branches of the synthesized program arise from iterations of Algorithm~\ref{algo:prog-synth}.
In the first iteration, \textsc{RSolve} returns a rational model
of $\varphi$ with, for example,
\[
x=\frac18,\qquad y=\frac12,\qquad z=-3.
\]
Algorithm~\ref{algo:prog-synth} then considers two single-output projections
obtained from this model.

First, it fixes $z=-3$ and leaves $y$ as the output to be synthesized.  Thus,
we obtain the single-output specification
\[
\varphi(x,\lambda,-3).
\]
The precondition for this branch is computed by quantifier elimination from
\[
\exists \lambda.\,\varphi(x,\lambda,-3).
\]

This corresponds to Lines~3--10 of Algorithm~\ref{alg:synth-program-example}.  Indeed, after substituting $z=-3$,
the first conjunct of $\varphi$ contains
\[
x^2+\frac{\lambda^2}{9}+\frac{9}{16}\le 1,
\]
and hence a solution for $\lambda$ can exist only when
\[
x^2\le \frac{7}{16},
\]
i.e. when $16x^2-7\le 0$.  Inside this branch, Algorithm~\ref{algo:qfnqaone}
synthesizes a rational value for $y$, and the program returns $(y,-3)$.
Note that the guard of each branch is the precondition $\psi_i(x)$ computed by quantifier elimination for the resulting single-output instance. 
Inside the branch, the remaining output variable is computed using \qepcad.

Then, from the same model, the algorithm fixes $y=\frac12$ and leaves $z$ as
the output to be synthesized.  This gives the single-output specification
\[
\varphi(x,\tfrac12,\lambda).
\]
Equivalently, using the model value $z=-3$ as the base point, this corresponds
to invoking \qepcad\ on
\[
\exists\lambda.\,\varphi(x,\tfrac12,\lambda-3),
\]
The resulting precondition is
\[
36x^2-35\le 0.
\]
This gives the guard in Lines~11--18 of
Algorithm~\ref{alg:synth-program-example}; inside the branch,
Algorithm~\ref{algo:qfnqaone} synthesizes the corresponding value of $z$.

After these two branches are generated, Algorithm~\ref{algo:prog-synth} updates
the residual specification by excluding the input regions already covered by the
computed preconditions.  Thus the next call to \textsc{RSolve} is made on
\[
\varphi(x,y,z)
\wedge \neg(16x^2-7\le 0)
\wedge \neg(36x^2-35\le 0).
\]
This residual formula is still satisfiable.  In the run shown here, \textsc{RSolve}
returns another rational model, for example
\[
x=-\frac{127}{128},\qquad y=0,\qquad z=\frac18.
\]

From this model, the algorithm again constructs single-output projections.
First, it fixes $z=\frac18$ and leaves $y$ as the output variable, obtaining
\[
\varphi(x,\lambda,\tfrac18).
\]
Quantifier elimination on
\[
\exists\lambda.\,\varphi(x,\lambda,\tfrac18)
\]
returns the precondition
\[
1024x^2-1023\le 0.
\]
This is the guard in Lines~19--26 of
Algorithm~\ref{alg:synth-program-example}.  Thus, this branch is introduced
because the residual specification after the previous iteration still had
rational models, and the model found by \textsc{RSolve} had $z=\frac18$.

Finally, from the same model, the algorithm fixes $y=0$ and leaves $z$ as the
output variable.  Equivalently, using $z=\frac18$ as the base point, this
corresponds to the single-output query
\[
\exists\lambda.\,\varphi(x,0,\lambda+\tfrac18).
\]
\qepcad\ returns
\[
(x+1\ge 0)\wedge (x-1\le 0),
\]
i.e. the precondition $-1\le x\le 1$.  This gives the guard of the final branch,
where Algorithm~\ref{algo:qfnqaone} synthesizes the corresponding value of $z$.

After adding these branches, the residual specification becomes
\[
\begin{aligned}
\varphi(x,y,z)
&\wedge \neg(16x^2-7\le 0)
 \wedge \neg(36x^2-35\le 0)\\
&\wedge \neg(1024x^2-1023\le 0)
 \wedge \neg\big((x+1\ge 0)\wedge(x-1\le 0)\big).
\end{aligned}
\]
This formula is unsatisfiable.  Hence the accumulated preconditions cover all
rational inputs for which the original specification is satisfiable.  Therefore,
the decision-list program in Algorithm~\ref{alg:synth-program-example} realizes
$\varphi$: it returns a rational pair $(y,z)$ satisfying $\varphi(x,y,z)$ whenever
one exists, and returns $\bot$ otherwise.

%% file: app-complete.tex
\section{Appendix for Section~\ref{sec:algorithm}}~\label{sec:app-complete}

\subsection{Solutions of \probreal\ are also solutions of \prob}~\label{app:sec:nra-synth}

\thmnrasynthesizable*

\begin{proof}
  Suppose algorithm ${\mathcal A}_{\mathcal L}$ generates a program
  $\prog \in \gram$ that realizes $\varphi(X,Y)$ over reals. Since
  ${\mathcal A}_{\mathcal L}$ is sound and complete for the fragment
  ${\mathcal L}$, it follows that for every input $\AAA \in
  \mathbb{R}^{|\XX|}$, $\prog(\AAA)$ generates a real-valued vector
  $\BBB \in \mathbb{R}^{|\YY|}$ such that $\varphi(\AAA,\BBB)$ holds.
  Since the assignments and guards of $\prog$ (which comes from \gram)
  use only polynomial terms with rational constants, it computes
  outputs as polynomial expressions over the inputs.  Since rationals
  are closed under polynomial maps, it follows that for every rational
  input $\AAA \in \mathbb{Q}^{|X|}$, if $\prog(A)$ terminates with
  output $\BBB$, then $\BBB \in \mathbb{Q}^{|Y|}$ as well. Since
  algorithm ${\mathcal A}_{\mathcal L}$ is complete for the fragment
  ${\mathcal L}$ over reals, and since $\mathbb{Q} \subset
  \mathbb{R}$, it follows that it is complete over rational inputs as
  well.  However, we have just argued above that when fed rational
  inputs, the synthesized program must generate rational outputs or
  $\bot$.  Hence ${\mathcal A}_{\mathcal L}$ is complete for the
  fragment ${\mathcal L}$ over rationals.

  For the second part of the theorem, consider the degenerate
  \qfnra\ fragment $\mathcal{L}$ consisting only of the specification
  $y^2 = x$.  Since this is a single output specification, we know
  from Section~\ref{sec:single-out} that there exists a sound and
  complete algorithm for solving \prob\.  However, if there was a
  sound and complete algorithm for solving \probreal\ (using grammar
  \gram) for this fragment ${\mathcal L}$, then this algorithm must
  give a program that takes any non-negative real value of $x$ as
  input, and gives its square-root as a real valued output, while only
  using assignments and conditions that are polynomials.  This means
  that for $x = 2$, it must compute $\sqrt{2}$ using polynomials over
  rationals.  Since rationals are closed under polynomial maps, and
  since $\sqrt{2}$ is irrational, this is an impossibility.  Hence,
  there is no sound and complete algorithm for solving $\probreal$
  over reals for the fragment ${\mathcal L}$.
\end{proof}

\subsubsection{{\qfnqa} with one output variable}\label{sec:1-var}

Algorithm~\ref{algo:qfnqaone} shows the pseudocode of $\prog(\inVar)$ that realizes $\varphi(\inVar, y)$ in $\qfnqa$. 
This algorithm uses a function ${\candsol}(\cdot)$ that takes as argument a polynomial equation and returns the finite set of candidate rational roots of the equation, as described above.
Further the algorithm also uses the function {\algqfnras} which when given a set of strict polynomial inequalities, returns a finite set of rational points, such that one of the points satisfies all the polynomial inequalities. 
This is implemented using the set of intervals computed using Real-root isolation~\cite{RRI-Uspensky,Collins-RRI,RRI-latest}.

Let $p$ be the the product of all the polynomials used in the set of constraints $\varphi(\inVar,\outvar)$ and consider the constraint of the form $p = 0$.
Clearly, there is a bijection between the roots of polynomials and the roots of $p$.
Then, applying real-root isolation (with a parameter $\epsilon$ for the size of the interval) to $p$, we get a set of disjoint intervals of the form
$I = \{(l_1,u_1), (l_2,u_2), (l_3,u_3), \cdots, (l_k,u_k)\}$.
We remark that a value of $\epsilon$ that ensures that the intervals are disjoint is computed by starting with a fixed value of $\epsilon$ and refining it until the intervals are disjoint. 
Then, the set of points returned by {\algqfnras} is of the form 
$\{l_1 -1, \frac{u_1+l_2}{2}, \frac{u_2+l_3}{2}, \cdots, \frac{u_{k-1}+l_k}{2},u_k + 1\}$.

\thmqfnqaone*
\begin{proof}
  The proof relies on the observation that every rational solution of
  $\varphi(\inVar,y)$ either satisfies $\widehat{\varphi}(\inVar, y)$
  or satisfies at least one equation $p = 0$, where $p \in
  \nstrict(\varphi)$. 
  The proof now follows directly from the 
  fact that {\algqfnras} is a pre-determined procedure for computing rational sample points for formulas in $\qfnra$ with strict constraints and Theorem~\ref{thm:rrt}.
\end{proof}

Finally, we remark that the procedure $\progone{\varphi}$ synthesized using Real Root Isolation (RRI) and the Rational Root Theorem (RRT) can be implemented as programs derivable from our grammar $\gram$.

  First, we consider the part of the program that uses the Rational Root Theorem. 
  Note that, given a univariate polynomial with integer coefficients, the set of candidate rational roots is finite and effectively computable from the integer factors of the constant and leading coefficients. Since this enumeration is finite and each candidate can be checked using polynomial equalities and inequalities, the entire procedure can be expressed using a finite composition of assignments and conditionals, all allowed in $\gram$.

  Next, we consider the part of the program that uses Real Root Isolation.
  Real-root isolation is a technique used to compute all the real roots of a univariate polynomial.
  For a given polynomial, RRI techniques output disjoint intervals on the number line such that each interval contains exactly one real root, where we can specify the size of the interval, denoted by $\epsilon$.
  Applying real-root isolation (with a parameter $\epsilon$ for the size of the interval) to the polynomial, we get a set of intervals of the form
  $I = \{(l_1,u_1), (l_2,u_2), (l_3,u_3), \cdots, (l_k,u_k)\}$.
  Note that the parameter $\epsilon$ is iteratively refined until the intervals are guaranteed to be disjoint.

  Finally, note that the only operation that is used in the RRI and RRT procedures outside our grammar is division, which can be implemented using repeated subtraction, conditionals and loops, and therefore can be expressed using $\gram$.

\section{Appendix for Section~\ref{sec:experiments}}~\label{sec:app-experiments}

%% file: app-benchmarks.tex
\clearpage
\phantomsection
\label{app:expt-results}
\addcontentsline{toc}{section}{Appendix for Section~\ref{sec:experiments}}

\includepdf[pages=-]{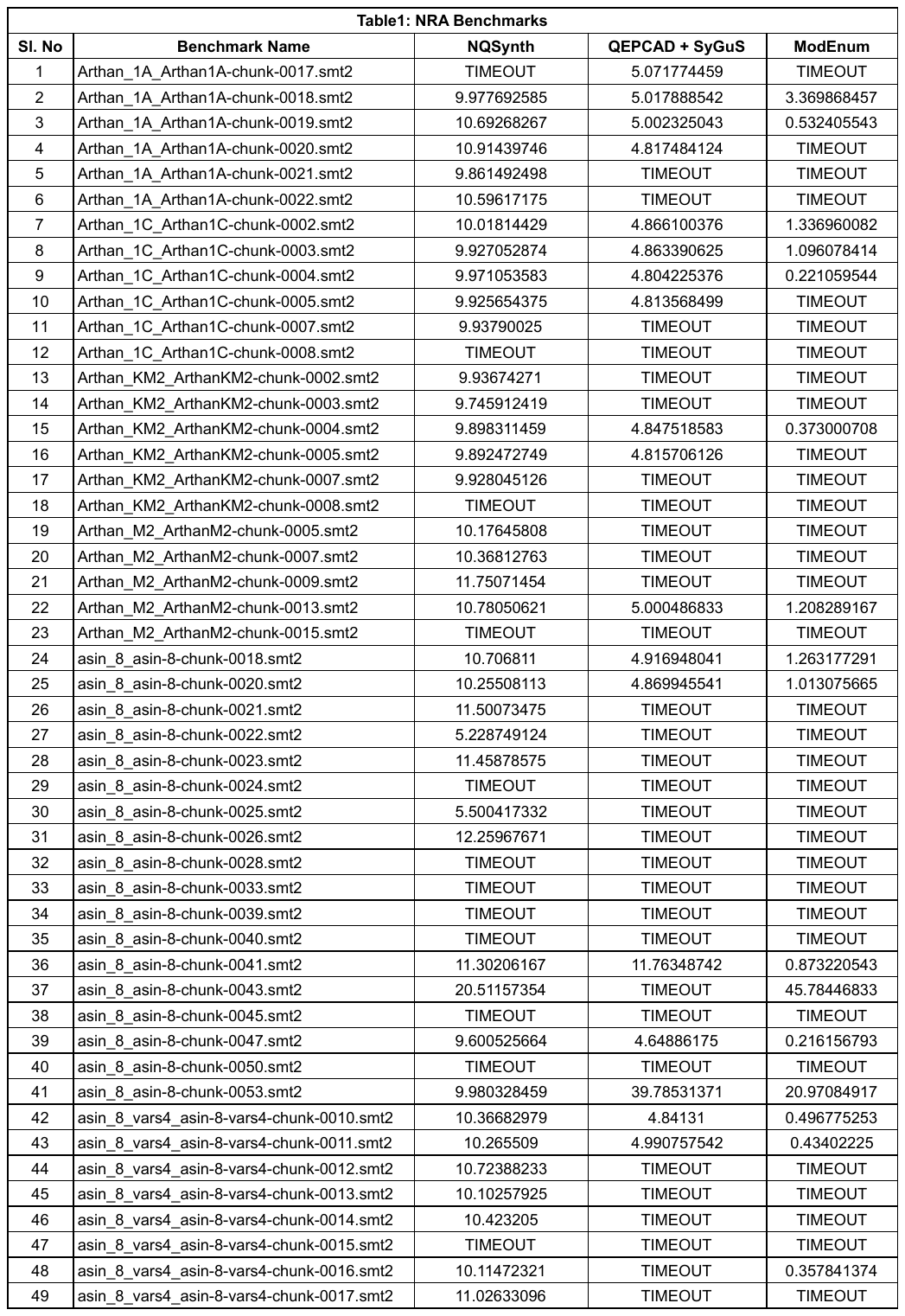}

\includepdf[pages=-]{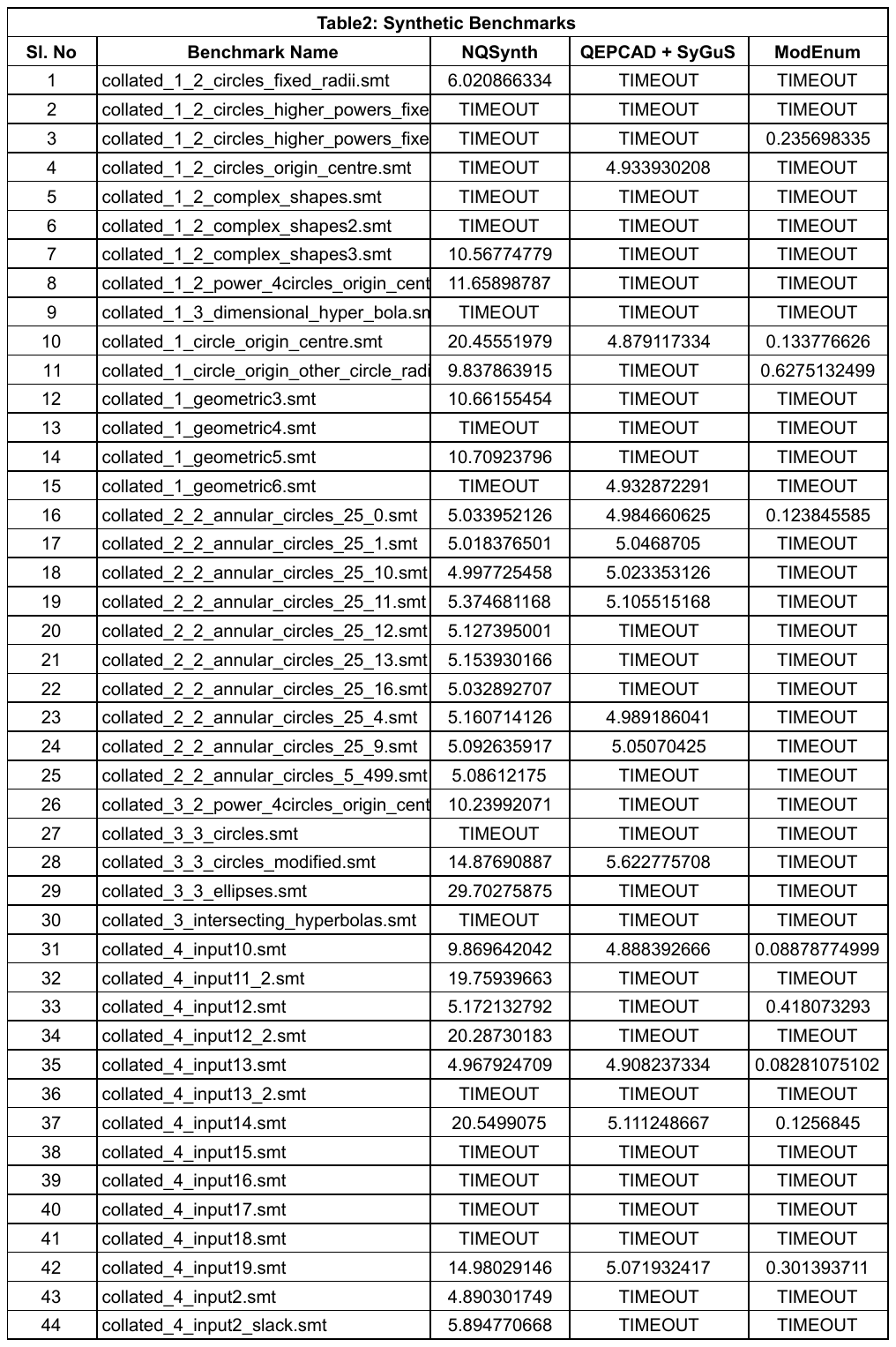}

\begin{table}[htbp]
\centering
\caption{Total Execution Time Comparison Across QE Backends (Timeout = 120s)}
\label{tab:qe_runtimes}
\begin{tabular}{lrrr}
\toprule
\textbf{Benchmark Filename} & \textbf{Sage (s)} & \textbf{MAPLE (s)} & \textbf{Mathematica (s)} \\ 
\midrule
\multicolumn{4}{l}{\textit{NRA Benchmarks}} \\
\texttt{Arthan\_1A\_Arthan1A-chunk-0017.smt2} & 119.73$^{\dagger}$ & 119.79$^{\dagger}$ & 23.91 \\
\texttt{Arthan\_1A\_Arthan1A-chunk-0020.smt2} & 11.87 & 1.39 & 1.21 \\
\texttt{Arthan\_KM2\_ArthanKM2-chunk-0002.smt2} & 9.63 & 0.72 & 0.31 \\
\texttt{Arthan\_KM2\_ArthanKM2-chunk-0008.smt2} & 119.97$^{\dagger}$ & 118.81$^{\dagger}$ & 119.32$^{\dagger}$ \\
\texttt{asin\_8\_asin-8-chunk-0023.smt2} & 10.04 & 1.05 & 0.88 \\
\texttt{asin\_8\_asin-8-chunk-0024.smt2} & 119.74$^{\dagger}$ & 119.14$^{\dagger}$ & 118.25$^{\dagger}$ \\
\texttt{asin\_8\_vars4\_asin-8-vars4-chunk-0015.smt2} & 119.81$^{\dagger}$ & 116.88$^{\dagger}$ & 118.14$^{\dagger}$ \\
\texttt{asin\_8\_vars4\_asin-8-vars4-chunk-0017.smt2} & 11.03 & 1.34 & 1.14 \\
\texttt{zankl\_gen-04.smt2} & 5.49 & 0.18 & 0.17 \\
\texttt{zankl\_gen-10.smt2} & 5.41 & 0.18 & 0.17 \\ 
\midrule
\multicolumn{4}{l}{\textit{Synthetic Geometric Benchmarks}} \\
\texttt{1\_2\_circles\_fixed\_radii.smt} & 6.02 & 0.28 & 0.15 \\
\texttt{1\_geometric3.smt} & 10.66 & 0.19 & 0.17 \\
\texttt{4\_input11\_2.smt} & 19.76 & 0.40 & 0.40 \\
\texttt{4\_input2.smt} & 4.89 & 0.13 & 0.14 \\
\texttt{4\_kissing\_set\_relax.smt} & 119.94$^{\dagger}$ & 117.07$^{\dagger}$ & 103.79$^{\dagger}$ \\ 
\bottomrule
\multicolumn{4}{l}{\small $^{\dagger}$ indicates a timeout.}
\end{tabular}
\end{table}